\documentclass[aps,pra, twocolumn]{revtex4-1}
\usepackage{amsthm, comment}
\usepackage[utf8]{inputenc}
\usepackage[sc,osf]{mathpazo}
\usepackage{graphicx,subfigure}
\usepackage{bm}
\usepackage{makeidx} 
\usepackage{amsmath}
\usepackage{amssymb}
 \usepackage{color, xcolor}
\usepackage{amsthm, MnSymbol}
\usepackage{hyperref}
\usepackage{cleveref}
\usepackage{diagbox}
\hypersetup{colorlinks, citecolor=blue, filecolor=blue, linkcolor=blue, urlcolor=blue}
\usepackage{cancel}
\usepackage[normalem]{ulem}
\usepackage{float}
\usepackage{multirow}
\usepackage{amsfonts}
\usepackage{mathtools}
\usepackage{derivative}
\usepackage{braket}
\usepackage{physics}

\newtheorem{theorem}{Theorem}

\newcolumntype{P}[1]{>{\centering\arraybackslash}p{#1}}
\newcommand{\ph}[1]{{\color{black}#1}}

\begin{document}

\title{ Entanglement of weighted graphs uncovering transitions in variable-range interacting models }

\author{Debkanta Ghosh, Keshav Das Agarwal, Pritam Halder, Aditi Sen(De)}

\affiliation{Harish-Chandra Research Institute, A CI of Homi Bhabha National Institute,  Chhatnag Road, Jhunsi, Allahabad - 211019, India}

\begin{abstract}

The cluster state acquired by evolving the nearest-neighbor (NN) Ising model from a completely separable state is the resource for measurement-based quantum computation.  Instead of an NN system, a variable-range power-law interacting Ising model  can generate  a  genuine multipartite entangled weighted graph state (WGS)  that may reveal intrinsic characteristics of the evolving Hamiltonian. We establish that the pattern of generalized geometric measure (GGM) in the evolved state with an arbitrary number of qubits is sensitive to fall-off rates and the range of interactions of the evolving Hamiltonian.  We report that the time-derivative and time-averaged GGM at a particular time can  detect  the transition points present in the fall-off rates  of the interaction strength,  separating different regions, namely long-range, quasi-local and local ones in one- and two-dimensional lattices with deformation. Moreover, we illustrate that in the quasi-local, and local regimes there exists a minimum coordination number in the evolving Ising model for a fixed total number of qubits which can mimic the GGM of the long-range model. In order to achieve a finite-size subsystem from the entire system, we design a local  measurement strategy that allows a WGS of an arbitrary number of qubits  to be  reduced to a local unitarily equivalent WGS having fewer qubits with modified weights.

\end{abstract}
\maketitle
\section{Introduction}
\label{sec:intro}

 The computational time of certain mathematical problems like prime factorization of integers \cite{factor} can be reduced utilizing the  quantum-mechanical principles compared to  classically available algorithms -- an important milestone that establishes the significance of building quantum computers \cite{nielsen_chuang_2010}. Since then, quantum speedups that cannot be achieved by the best classical computer have been demonstrated in distinguishing classes of functions \cite{deutsch1992,Bernstein1997} in a search problem from databases \cite{grover1997} in estimating phases of an operator \cite{phase95,kitaev02} and for solving linear systems of equations \cite{hhl09} to name a few. Although entanglement \cite{horo_ent_09} has been shown to be beneficial for accomplishing higher efficiency in certain quantum protocols, including various computational tasks \cite{linden01,jozsa03,nest13} and quantum communication \cite{cryp_1991,bennet1992,cryp_1992,tele_1993,dc1996,cryp_2000,cryp_2002,BENNETT20147,Ren2017,guo2019}, the resource required to attain quantum supremacy in quantum algorithms has yet to be determined \cite{ahnefeld22,naseri22}.

On the other hand, entanglement is one of the prerequisites for measurement-based quantum computation (MBQC) \cite{Rauss2001, Hein2004,hein2006entanglement,casati2006quantum,browne2006oneway,briegel2009measurement}  which allows one to construct all  the universal quantum gates required in a quantum computer \cite{divincenzo95}. In the paradigm of MBQC or \textit{one-way quantum computation} \cite{Rauss2001,browne2006oneway}, generation of genuine multipartite entangled (GME) states and, more specifically, multipartite cluster states \cite{briegel01, Yokoyama2013} remains crucial.
These states are  specific types of stabilizer states, known also as graph states, which are employed as a universal resource for  measurement-based circuits followed by Clifford operations, and they are easy to simulate through classical computers \cite{scott_2004_pra, simon_2006_pra,nielsen_chuang_2010}. The corresponding states have been successfully generated on several physical platforms  and are used to show a variety of information processing tasks \cite{zhang06,Prevedel2007,tame07,lanyon13,Bell_2013,Yokoyama2013,bell_14,bell2_14,mamaev19}. 
Beyond one-way models, numerous kinds of universal resource states, such as two-dimensional ($2$D) cluster \cite{briegel01,Brell_2015}, Affleck-Kennedy-Lieb-Tasaki-type states \cite{aklt88}, Toric codes \cite{KITAEV20032}, and weighted graph states (WGSs) \cite{ahufinger06,hartmann_briegel_2007,Anders_2007,peng12}, are also shown to be useful for quantum computation \cite{beyond1way07,grossnovel07}.

The original proposal of producing cluster states, which are necessary for MBQC,  solely employs the nearest-neighbor (NN) Ising model \cite{briegel01}. However, long-range (LR) quantum spin models are naturally created by neutral atoms in optical lattices interacting via dipole interactions \cite{Lewenstein2007,Sowiński2012} or trapped ions \cite{iontrap08,iontrap09,iontrap12,Islam_2011,Britton_2012} under Coulomb potential \cite{iontrap_cirac04}.   Furthermore, such LR models frequently exhibit critical phenomena and possess phases that are not captured by short-range (SR) quantum spin models
\cite{koffel2012,vodola2014,lr_ising_vodola_2016,lr_ising_defenu_2020,lr_model_romanroche2023,vodola23_prb}. For instance, counter-intuitive results in LR models include  breaking of the Mermin-Wagner-Hohenberg theorem \cite{Mermin1966,Hohen1967,peter2012}, the rapid propagation of correlations \cite{Maghrebi2016,gong2017,ares2018,Ares_2019}, and violations of entanglement area law \cite{cramer2010,koffel2012,scha2013,cadarso2013}. 
In addition, it has been demonstrated that  on one hand, the ground or thermal states of the LR model possess a high multipartite entanglement or more spreading of  bipartite entanglement and, on the other hand, multipartite entangled states can also be  created through the LR models \cite{koffel2012,vodola2014,ren20,lakkaraju2020freezing,lgcl_asd_pla2021,lgcl_asd_pra2022,francica_22_prb, gong_cirac_23_prl}.

In this context,  it is interesting to explore the properties of the dynamical state obtained via interacting LR spin systems and its suitability for MBQC, starting from a suitable product state. In particular, investigations have focused on the scaling of block entanglement, two-point correlation functions, multipartite entanglement \cite{aditi_pra_2010,biswas14}, and quantum discord \cite{olliver2001,Modi2012,Bera_2018} of the state following evolution under the LR Ising model \cite{dur_hartmann_prl_2005,mahto_shaji_pra_2022}. 
Moreover, random quantum circuit implementation by carrying out measurements on weighted graph states \cite{ran_circuit_by_wg08} and robust entanglement concentration protocols \cite{economou2023} for generating high-fidelity Greenberger–Horne–Zeilinger states \cite{Greenberger2007}  have been proposed in the presence or absence of coherent and incoherent noise. 

The genuine multipartite entanglement \cite{meyer02,hashemi12,sadhukhan17,Gour2018monogamy,xie21} of the dynamical state in various geometries and its characteristics under local measurements have not yet been addressed in the context of variable-range (VR) interactions included in the evolution Hamiltonian for one-dimensional $(1D)$ and $2$D lattices. On one hand, these investigations can highlight the potential of the evolved state as a resource in quantum information protocols and identify phases  and critical phenomena that are present in the evolution Hamiltonian. On the other hand, the impact of local measurements can indicate how to decouple certain parts of the circuits from the entire circuits, which can play a vital role on the performance of the computation. To accomplish this goal, we first show that the evolved state generated via VR interactions, referred to as a weighted graph state,  is GME having nonvanishing generalized geometric measure \cite{wei03} (GGM) \cite{aditi_pra_2010} and its pattern with time depends on the fall-off rate of interactions and the coordination number.  
Moreover, we know that the power-law fall-off in interactions of the LR transverse Ising model undergoes transitions from long-range to quasi-local and from quasi-local to local ones \cite{koffel2012,vodola2014,lr_ising_vodola_2016,lr_ising_defenu_2020,lr_model_romanroche2023,vodola23_prb}. We report that the time-derivative and time-averaged GME content of the WGSs  can detect transitions, present in the fall-off rates in one- and  two-dimensional lattices. More importantly, we demonstrate that if the $2$D square lattice is distorted with an arbitrary angle, resulting in a hexagonal lattice, the discontinuity in the derivative of the multipartite entanglement with respect to time or fall-off rate can predict the transition points in the fall-off rate of the evolving Hamiltonian. Further, we determine that there exist  threshold values on the total number of qubits and coordination number above which the GME state produced via VR interactions remains constant, thereby providing  lower bounds on the number of two-qubit gates required to simulate the interacting Hamiltonian used for producing WGSs. The results can be important during the realization of the WGS in laboratories e.g., with superconducting qubits \cite{Song2017,ParallelSong2017,Pedersen2019,Song2019}.

Certain WGSs with a minimal number of qubits may be necessary during the implementation of certain algorithms, and in such cases, the goal will be to generate WGSs with a required number of qubits from the same state with an arbitrary number of qubits by making local measurements. For cluster states, such a measurement approach is well known \cite{rauss_pra2003}. In this paper, we present a local measurement strategy that allows us to  obtain a local unitary equivalent WGS with adjusted weights.

The paper is structured in the following manner. In Sec. \ref{sec:Weighted Graph State}, we introduce the exact expression  of the WGS in various lattice geometries. The closed forms of GGM are calculated in Sec. \ref{sec:GGMWG}. Using the trends of GGM, detection of transitions in the fall-off rates of the evolution Hamiltonian  is presented in Sec. \ref{sec:alpha_indicator} while the investigations of emulating the GME state in the LR  model through the SR  ones is examined in Sec. \ref{sec:mimick}. We discuss the impact of local measurements on the WGS in Sec. \ref{sec:measurement_effect}.  Concluding remarks are included in Sec. \ref{sec:conclu}.

\section{Generation of A Weighted Graph State via A long-range Hamiltonian with varying interaction strength}
\label{sec:Weighted Graph State}

A weighted graph state of $N$ parties can be expressed by an underlying graph $G=(V,E)$, with the $N$ parties forming the set of vertices, $V=\{1,2,\ldots,N\}$ connected by a set of edges $E\subset V\times V$ based on the interactions among the parties. The adjacency matrix of $G$ is $\Phi_E=(\phi_{ij})$, where the weight $\phi_{ij}$ denotes the interaction between the parties (vertices) $i$ and $j$. We consider real weights, i.e., $\Phi_E$ is a real and symmetric matrix. Therefore the underlying graph $G$ is simple and undirected, i.e. $\phi_{ii}=0$ and $\phi_{ij}=\phi_{ji}\forall (i,j)\in V\times V$.
\begin{figure*}
\centering
    \includegraphics[width=\linewidth]{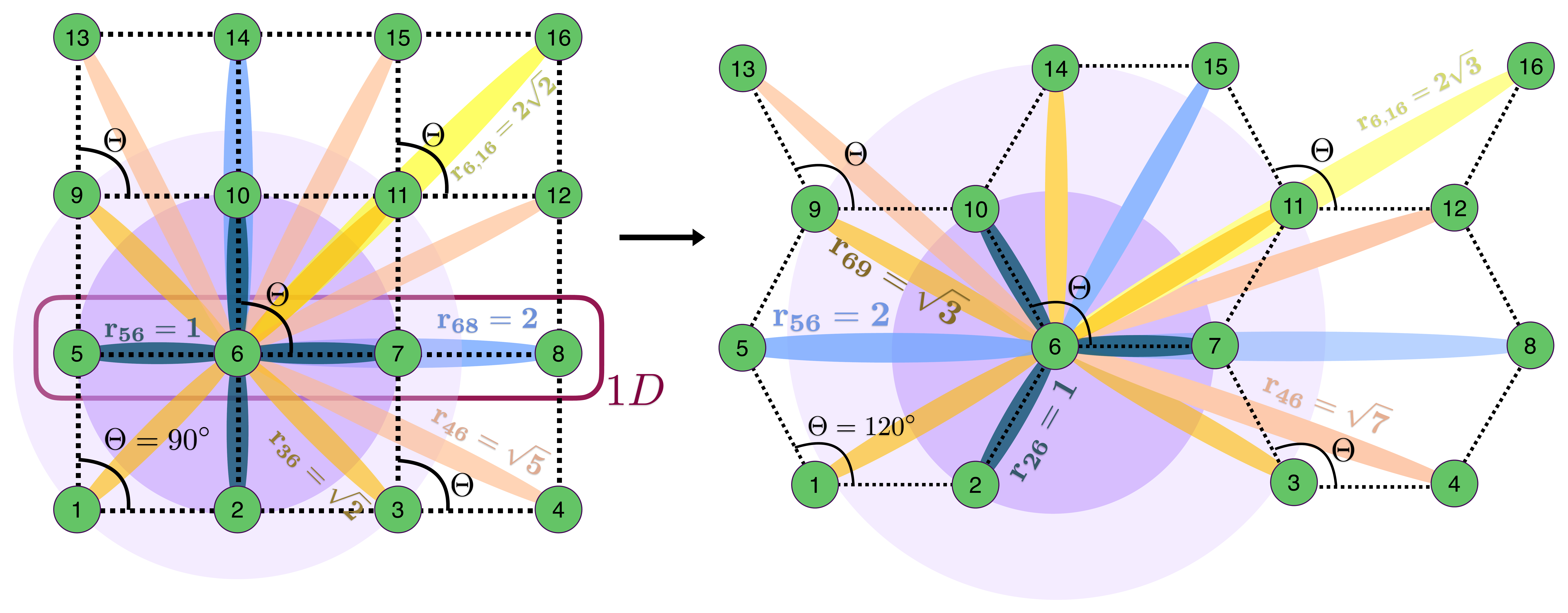}
    \caption{The two-dimensional lattice with sites $L\times L$ ($L=4$) can be  transformed from the square (left) to the honeycomb (right) lattice, forming a fully connected weighted graph state with different $\Theta$ in open boundary condition. Connections are shown for a single qubit with the corresponding distances. $\Theta=90^\circ$ and $120^\circ$ correspond to the square and the honeycomb lattices respectively. Each site $\va{i}=(i_x,i_y)$ is indexed as $i=(i_x-1)*L+i_y$, where $i_x\in[1,L]$, $i_y\in[1,L]$, and $i_x,i_y\in\mathbb{N}$ . Different colors (shades) signifiy different strengths of the interactions, which decreases with increasing distance.}
    \label{fig:schematic3}
\end{figure*}

To construct a spin-$\frac{1}{2}$ graph state, each vertex is a qubit initialized in $\ket{+}=(\ket{0}+\ket{1})/\sqrt{2}$, where $\ket{0}$ and $\ket{1}$ are the eigenstates of the Pauli matrix $\sigma_z$ with the corresponding eigenvalues $1$ and $-1$. Furthermore, the two-qubit interacting unitary, $U_{ij}=e^{-i H_{ij}t} \forall i,j\in V$, acts on each vertex, entangling the pairs of qubits (creating the edges), which is achieved by the two-qubit Hamiltonian,  $H_{ij}=\phi_{ij}\Big(\frac{1-\sigma^z_{(i)}}{2}\Big)\Big(\frac{1-\sigma^z_{(j)}}{2}\Big) $. Hence, the unitary transformation, in this case, takes the diagonal form as $U_{ij}=\text{diag}(1,1,1,e^{-i g_{ij}(t)})$, known as the controlled-phase gate with $g_{ij}(t)=\phi_{ij}t$. Here, $\phi_{ij}$, which comprises the information of the range and strength of the interaction, is constant throughout the time. Therefore, the total interacting Hamiltonian acting on the graph can be written as $H=\sum_{i=1}^{N-1}\sum_{j>i}^N H_{ij}$

Notice that the Hamiltonian remains unchanged under the exchange of two indices, $H_{ij}=H_{ji}$ and $[H_{ij},H_{jk}]=0$, which satisfies the criteria for the graph having undirected and unordered edges respectively. Finally, the graph state for $N$-party system can be expressed as 
\begin{eqnarray}
    \ket{\Psi_G(t)}_N &=&\prod_{i,j} U_{ij}(t) \ket{+}^{\otimes N}.
    \label{eq:wgs_def}
\end{eqnarray}

Considering only nearest-neighbor interactions in one dimension, i.e., $H=\sum_i H_{i,i+1}$, and setting $\phi_{i,i+1}t=(2n-1)\pi$ for any $n\in \mathbb{N}$, (with $\mathbb{N}$ being the set of natural numbers) which transforms the edge connecting unitaries to be $U_{i,i+1}=\text{diag}(1,1,1,-1)\equiv \text{controlled-NOT}$, we get the well-known MBQC resource state, known as the cluster state \cite{briegel01} having the form $\ket{\Psi_G((2n-1)\pi)}_N =\frac{1}{2^{N/2}} \bigotimes_{i=1}^N\big(\ket{0}_i+\ket{1}_i\sigma_{i+1}^z\big), $
using the convention that $\sigma_{N+1}^z=\mathbb{I}$ in the open boundary condition.

Instead of this simplified  scenario, we concentrate on the more general experimentally viable picture of all-to-all connectivity with power-law interaction strength, i.e.,
\begin{equation}
    \phi_{ij}=|i-j|^{-\alpha}=r_{ij}^{-\alpha},
    \label{eq:stre_open}
\end{equation}
where $\alpha \geq 0$ dictates the fall-off rate of the power-law interactions and $r_{ij}$ denotes the spatial distance between the lattice sites, $i$ and $j$ . For the one-dimensional spin system, $r=|i-j|$, whereas for a $2$D  square lattice $r=|\va{i}-\va{j}|$, which is the Euclidean distance between the site at $\va{i}$ and the site at $\va{j}$. For example, $r=1$ represents the nearest-neighbor NN sites in both $1$D and $2$D square lattices, while $r=2$ and $\sqrt{2}$ are the distances of the next-nearest neighbor (NNN) in $1$D and $2$D square lattices respectively (see Fig. \ref{fig:schematic3} for illustration). For $N=2$, one obtains the corresponding state, $\ket{\Psi(t)}_2 = \frac{1}{2}\big(\ket{00}+\ket{01}+\ket{10}+e^{-i g_{1 2}}\ket{11}\big)$, where $g_{12} = \phi_{12}t$. Similarly, for $N=3$, we have
\begin{eqnarray}
    \nonumber \ket{\Psi(t)}_3 &=& \frac{1}{2\sqrt{2}}\big(\ket{000}+\ket{001}+\ket{010}+e^{-i g_{2 3}}\ket{011} \\ \nonumber &+& \ket{100}+e^{-i g_{1 3}}\ket{101}+e^{-i g_{1 2}}\ket{110}\\ &+& e^{-i (g_{1 2}+g_{1 3}+g_{2 3})}\ket{111}\big),
    \label{eq:3qubit_all}
\end{eqnarray}
with each $g_{ij}=\phi_{ij}t$, defining the weights of the graph. From the above form of the states, we see that for the $N$-qubit state expressed in the $2^N$ computational basis, the coefficient of the basis vectors $\{\ket{\eta}=\ket{a_1 a_2 \ldots a_N}\}$ where $\eta$ is the decimal equivalent of $\ket{a_1 a_2 \ldots a_N}$ with $a_i= 0,1 \, \forall i$ turns out to be $c_{\eta}(t)=\exp\bigg(-i\sum_{i=1}^{N-1}\sum_{j>i}^N g_{i j}(t)\ph{ a_i a_j}\bigg)$. Therefore, the $N$-qubit weighted graph state  can be written as 
\begin{eqnarray}
    \nonumber \ket{\Psi_G(t)}_N &=&\frac{1}{2^{N/2}}\sum_{\eta=0}^{2^N-1}\exp\Big(-i\sum_{i=1}^{N-1}\sum_{j>i}^{N} g_{i j}(t)\ph{ a_i a_j}\Big)\ket{\eta} \\ &=& \frac{1}{2^{N/2}}\sum_{\eta=0}^{2^N-1}c_\eta(t)\ket{\eta}.
    \label{eq:wgs}
\end{eqnarray}
Interestingly, we will demonstrate that the coefficients of $c_{\eta}(t)$ carrying the signature of variable-range interactions of the system can reveal counter-intuitive properties which the conventional cluster state does not possess.

Let us consider $2$D lattices by introducing a distortion in the $2$D square lattice of size $L\times L$, such that for a site $\va{i}$, with the index $(i_x, i_y)$, $i_x\in[1,L]$, $i_y\in[1,L]$ with $i_x,i_y\in\mathbb{N}$, the distortion is as follow:\\
 \textit{Case I- $i_x+i_y$ is even}. Sites indexed $(i_x+1, i_y)$, $(i_x, i_y-1)$, and $(i_x, i_y+1)$ are the nearest-neighbor sites and the angle between vectors joining $(i_x, i_y)$ to $(i_x, i_y+1)$ and $(i_x, i_y)$ to $(i_x+1, i_y)$ is $\Theta\geq90^\circ$. \\
 \textit{Case II- $i_x+i_y$ is odd}. Sites indexed $(i_x-1, i_y)$, $(i_x, i_y-1)$, and $(i_x, i_y+1)$ are the nearest-neighbor sites and the angle between vectors joining $(i_x, i_y)$ to $(i_x, i_y+1)$ and $(i_x, i_y)$ to $(i_x-1, i_y)$ is $\Theta\geq90^\circ$.
 
As shown in Fig. \ref{fig:schematic3},  $\Theta=90^\circ$ represents the square lattice while $\Theta=120^\circ$ corresponds to the honeycomb lattice, where the lattice is tiled by regular hexagons. For arbitrary lattice angle $\Theta$, the $2$D lattice is deformed and is tiled by symmetric but non-regular hexagons with distance between one pair of parallel sides less than the distance between the other two parallel pairs of the hexagon, and the pair with the smallest side oriented in one axis ($x$-axis in Fig. \ref{fig:schematic3}). The lattice is scaled such that the nearest neighbors are separated by a unit distance.  Note that this deformation reduces the $\mathbb{Z}_4$ symmetry of the square lattice to $\mathbb{Z}_2$ symmetry for other lattice structures in Fig. \ref{fig:schematic3}.

\section{Genuine Multipartite entanglement in a Weighted Graph State}
\label{sec:GGMWG}

Let us analyze the multipartite entanglement content of the WGS originated from the variable-range interacting evolution operator. As we will illustrate, such investigations establish capability of the WGS as a resource for quantum information processing tasks and, at the same time, can uncover quantum features of the evolution Hamiltonian. To characterize its resource, we focus on its genuine multipartite entanglement content. A multipartite pure state is genuinely multipartite entangled  when it is not separable across any bipartition. Although numerous multipartite entanglement measures \cite{Blasone2008,
Or2008,Vidal2008,Rom2008,Wei2010, Eisert01, meyer02,hashemi12,sadhukhan17,Gour2018monogamy,xie21,wei03} are proposed which are shown to be crucial for quantum information protocols \cite{buzek97,hillery99,bruss04,Shi_2010}, they are not always easy to compute, specially for large system size. We choose the generalized geometric measure to quantify its GME content. The GGM, $\mathcal{G}$, is a distance-based measure, defined as the distance of a given state from the set of non-genuinely multiparty entangled states. For a given \(N\)-party pure state, $|\Psi\rangle$, GGM can be computed by the expression written as
\begin{equation}
    \mathcal{G}(|\Psi\rangle) = 1-\max_{A:B}\{\lambda^2_{A:B} | A\cup B = \{1,2,\dots N\}, A\cap B=\O \},
    \label{eq:ggm_def}
\end{equation}
where $A$ and $B$ are any arbitrary bipartitions of the multiparty state, and $\lambda_{A:B}$ is the maximum Schmidt coefficients in the $A:B$ bipartition, along with maximization over all possible bipartitions.

\subsection{Computation of reduced density matrices of WGSs}

From the definition of the GGM, we know that the largest eigenvalue among all the bipartitions of the given state  contributes to the computation of GGM.  We use the projected entangled-pair states (PEPS) description of the weighted graph states \cite{verstraete_pra_2004, dur_hartmann_prl_2005, hartmann_briegel_2007} to evaluate the reduced density matrices. We argue that the single-site reduced density matrix has maximal Schmidt coefficients among all possible density matrices.

For an all-to-all connected WGS, each qubit is acted upon by $N-1$ commuting unitaries and, therefore, a single qubit can be replaced by $N-1$ virtual qubits, each in the $\ket{+}$ state initially. The $\bar{k}_l$ virtual qubit in the $k$th site interacts with the $\bar{l}_k$ virtual qubit at site $l$ by the unitary $U_{kl}$. These virtual qubits now form $N(N-1)/2$ valence bond pairs, and each unitary acts on each valence bond pair independently, in the complex Hilbert space  $(\mathbb{C}^2)^{N(N-1)}$. Therefore, the weighted graph state in Eq. (\ref{eq:wgs_def}) can be written upto normalization as $\ket{\bar{\Psi}_G(t)}= \underset{\bar{k}_l \bar{l}_k}{\otimes}\ket{\zeta}_{\bar{k}_l \bar{l}_k} = \underset{\bar{k}_l \bar{l}_k}{\otimes}U_{kl} \ket{+}_{\bar{k}_l}\ket{+}_{\bar{l}_k}$ where 
\begin{equation}
    \ket{\zeta}_{\bar{k}_l \bar{l}_k} = \frac{1}{2}\big(\ket{0_{\bar{k}_l}0_{\bar{l}_k}}+\ket{0_{\bar{k}_l}1_{\bar{l}_k}}+\ket{1_{\bar{k}_l}0_{\bar{l}_k}}+e^{-i g_{kl}}\ket{1_{\bar{k}_l}1_{\bar{l}_k}}\big).
\end{equation}
The original state is now recovered upto normalization by local projection, $P_k = \ket{0}_k\bra{0_{\bar{k}_1}0_{\bar{k}_2}\dots0_{\bar{k}_{N-1}}}+\ket{1}_k\bra{1_{\bar{k}_1}1_{\bar{k}_2}\dots1_{\bar{k}_{N-1}}}$. i.e., for each original qubit, the $N-1$ virtual qubit is now projected back to the two-dimensional Hilbert space, where only the completely polarized states of virtual qubits (all $\ket{0}$ and all $\ket{1}$) at a particular site are projected out. Finally, the weighted graph state as given in Eq. (\ref{eq:wgs}) is recovered.
An all-to-all connected WGS with the underlying graph as $G=(V,V\times V)$ has the adjacency matrix $\Phi_E=(\phi_{ij})$ defining the controlled-phase gates $U_{ij}$.  

For the state $\displaystyle \ket{\Psi_G}=\Big(\prod_{(i,j)\in V\times V} U_{ij}\Big) \Big(\bigotimes_{k\in V}\ket{+}_k\Big)$, an arbitrary bipartition is done by dividing the sites into two subsystems $A$ and $B$, such that $A\cup B=V$ and $A\cap B=\emptyset$. To write the local density matrix of a subsystem $A$, we take the partial trace of the density matrix of the whole system over the subsystem $B$. Therefore, the unitaries which act only on the sites in the subsystem $B$, i.e., $U_{ij}$ such that $i,j\in B$, annihilate, due to their commuting nature and the cyclic property of the trace. We can express the reduced density matrix of the subsystem $A$ as
\begin{eqnarray}
    \rho_{A}&=&\sum_{k,k^\prime\in A}U_{k,k^\prime}\trace_B[\ket{\Psi^\prime}\bra{\Psi^\prime}]U^\dagger_{k,k^\prime}=\sum_{k,k^\prime\in A}U_{k,k^\prime}\rho^\prime_{A}U^\dagger_{k,k^\prime},\nonumber\\
    \ket{\Psi^\prime}&=&\bigotimes\limits_{k\in A,l\in B}^{}U_{k,l}\ket{+}_k\ket{+}_l.
    \label{eq:psi_reduced}
\end{eqnarray}
Since the unitaries $U_{k,k^\prime}$ act on the subsystems of $A$, so the eigenvalues of $\rho_A$ and $\rho_A^\prime=\trace_B[\ket{\Psi^\prime}\bra{\Psi^\prime}]$ are the same. With a subsystem $A$ of $n$ sites ($1\leq n<N$), the underlying graph of $\ket{\Psi^\prime}$ is $G^\prime\subset G$ with only the edge where each site of $A$ is connected to each site of $B$, contributing to the eigenvalues. Computing the reduced density matrices $\rho_A^\prime$ by using $\ket{\bar{\Psi}^\prime}$ which is $\ket{\Psi^\prime}$ in the PEPS formalism, involves three steps to keep track of indices in Eq. (\ref{eq:psi_reduced}).
\begin{enumerate}
    \item Projection on virtual qubits in $B$, $\ket{\Psi^P}=P_B\ket{\bar{\Psi}^\prime}=\bigotimes\limits_{l\in B}^{}\left[\prod_{k\in A}U_{kl} \ket{+}_{\bar{k}_l}\ket{+}_l\right]$.
    \item Partial trace over $B$, $\bar{\rho}_A = \bigotimes_{l\in B}\bar{\rho}_A(l)$, where each term is described in Eq. (\ref{eq:rho_l}) in  Appendix \ref{sec:WGS_PEPS}. 
    \item Projection on virtual qubits in $A$, $\rho^\prime_A=\underset{l\in B}{\ostar}\bar{\rho}_A(l)$, where $\ostar$ is the element-wise product of matrices, known as the Hadamard (\textit{Schur}) product.
\end{enumerate}

These steps are further explained in  Appendix \ref{sec:WGS_PEPS}. For the following discussion, we use the normalized form of the final Hadamard product form of $\rho^\prime_A$ as $\rho_A$.

The above procedure can be easily followed for computing single-site reduced density matrices for $A_k=\{k\}$ i.e., $|A_k|=n=1$. The $k$th single-site reduced density matrix of the system can be written as $\rho_k = \underset{l\in B}{\ostar}\rho_k(l)$ where  $\rho_k(l)=\frac{1}{2}\left(\ket{+}_k\bra{+}_k+\ket{\chi_l}_k\bra{\chi_l}_k\right)$, with $\ket{\chi_l}_k=\frac{1}{\sqrt{2}}\left(\ket{0}_k+e^{-i \phi_{kl}}\ket{1}_k\right)$. The closed (normalized) form  of  $\rho_k$ can be expressed as
\begin{eqnarray}
    \rho_k=\begin{pmatrix}
    \frac{1}{2} & \frac{1}{2^N}\prod_{\substack{j=1 \\ j\neq k}}^N(1+e^{i \phi_{kj}})\\
    \frac{1}{2^N}\prod_{\substack{j=1 \\ j\neq k}}^N(1+e^{-i \phi_{kj}}) & \frac{1}{2}
  \end{pmatrix}.
  \label{eq:single_site}
\end{eqnarray}

\subsection{Maximizing the Schmidt value of a weighted graph state}

By calculating GGM numerically upto $N=14$, we find that the maximum eigenvalue always comes from the single-site reduced density matrices as shown in Appendix \ref{sec:WGS_GGM}. Combining both the analytical and the numerical results, the closed form of GGM for the WGS, by diagonalizing the single-site density matrix in Eq. (\ref{eq:single_site}), can be given as
\begin{equation}
     \mathcal{G}(N,z,\alpha,t)=\frac{1}{2}-\frac{1}{2}\Bigg|\prod_{r=1}^z\cos\bigg(\frac{t}{2r^{\alpha}}\bigg)\Bigg|,
     \label{eq:ggm_open}
\end{equation}
where $z$ is the range of the interaction, which can also be referred to as the coordination number. 
 Note that for the all-to-all connected lattice, we have $1\leq z \leq N-1$.

\begin{figure}[h]
    \centering
    \includegraphics[width=\linewidth]{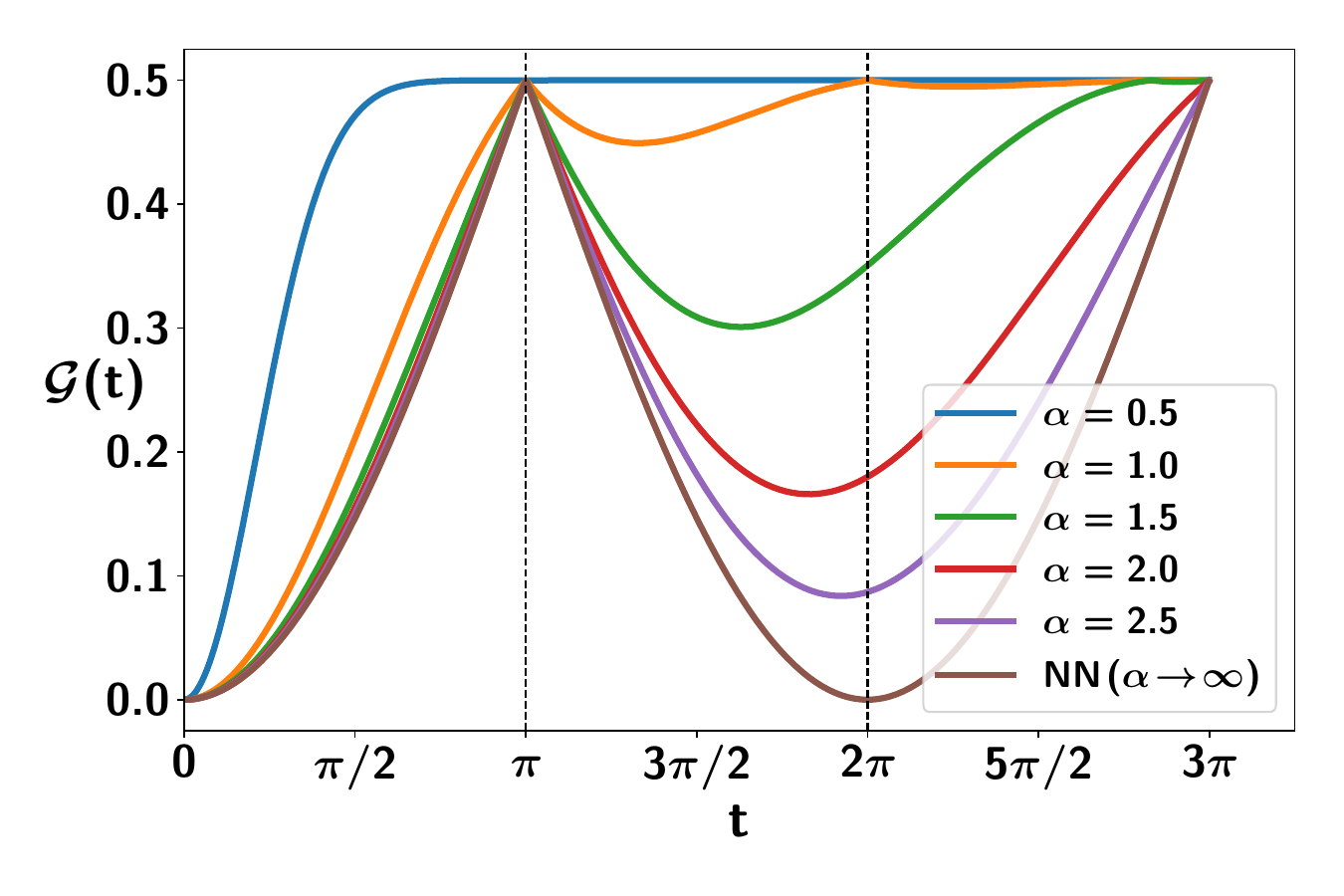}
    \caption{(Color online.) GGM, $\mathcal{G}$ (vertical axis) against time $t$ (horizontal axis) for the $N=5000$ all-to-all connected WGS ($z=N-1$) in a $1$D lattice with open boundary condition. Different lines represent different values of fall-off rate $\alpha$ (from below, \(\alpha\) increases). $\mathcal{G}$ reaches maximum value $0.5$ at $t=(2n+1)\pi \forall n\in \mathbb{Z}$ irrespective of $\alpha$ since  constant nearest-neighbor interactions are always present in all the situations. Both the axes are dimensionless.} 
    \label{fig:ggm_signal}
\end{figure}
{\bf Remark 1.} An unavoidable period exists in the weighted graph state if all the weights are rational numbers ($\mathbb{Q}$).
\begin{proof}
    The product of cosines can be written as a sum of cosines. For a WGS with the GGM contribution from the site $p$, with weights $J_{pj}$ of different sites $j$, GGM can be expressed as
    \begin{equation}
        \mathcal{G}(\{J_{pj}\},t) = \frac{1}{2}\left(1-\sum\limits_{perm \in \mathcal{S}}^{}\cos\sum_{i}\Big((-1)^{perm(i)}J_{pi}\Big)t\right),
        \label{eq:cos_sum}
    \end{equation}
    where $\mathcal{S}$ is a set of $N$-bit strings with the first bit as zero. Therefore, the cardinality of $\mathcal{S}$ is $2^{N-1}$. The total sum over $\mathcal{S}$ denotes the different possible summations that arise in the conversion of the product of cosines to the sum of cosines.

    The sum of cosines has a period which is the lowest common multiple of all the arguments in cosines. Therefore,  the weights $J_{pi}\in\mathbb{Q}$ and $\mathcal{G}$ is periodic in $t$, which in the case of $\alpha=1$ scales as $O(N!)$.
\end{proof}

{\bf Remark 2.}
The weighted graph state has no period if $\alpha$ is not an integer in one dimension.

\begin{proof}
    If $\alpha\notin\mathbb{N}$, $J_{pi}\notin\mathbb{Q}$, which follows from the Eq. (\ref{eq:cos_sum}), shows that $\mathcal{G}$ is aperiodic.
\end{proof}

\section{Recognizing transition in fall-off rates via the GGM of a Graph state }
\label{sec:alpha_indicator}

To create the weighted graph state with the help of the LR Ising Hamiltonian having $z=N-1$, the initial state is taken to be the equal superposition of all the elements in the computational basis, which is a fully separable  state having vanishing GGM. Suppose at time $t<0$, there is a strong magnetic field (with respect to the nearest-neighbor interaction strength) in the positive $x$-direction, causing the initial state to be completely polarised in the $\ket{+}$ state \cite{Kyaw_2018}. At $t>0$, sudden quench to the LR Ising Hamiltonian by abruptly turning off the magnetic field also leads to the WGS. The amount of multipartite entanglement generated in the resultant network depends upon the weight of the connection between the vertices, determined by the fall-off rate of the evolving Hamiltonian and the coordination number.

\subsection{Detection of transition in fall-off rate with a weighted graph state in one dimension}

Based on the LR system Hamiltonian and the underlying lattice, the fall-off rate $\alpha$ possesses various transition points, which separates the non-local (strong long-range order) from the quasi-local (weak-long range) order ($\alpha^*_{LR}$) and the quasi-local order from the local (short-range) one ($\alpha^*_{SR}$) in the ground state of the system. In the quantum spin models, analysis shows that $\alpha^*_{LR}=d$, where $d$ is the dimension of the lattice, exhibits the transition from long-range to quasi-local region \cite{koffel2012,vodola2014,lr_ising_vodola_2016,lr_ising_defenu_2020,lr_model_romanroche2023,vodola23_prb}. 

Let us demonstrate that the trends of multipartite entanglement in the WGS constructed from the underlying long-range model can capture the transition point in the fall-off rate $\alpha$.
 
\subsubsection{GGM of the weighted graph state in one dimension}
Let us first examine the behavior of GGM of the WGS generated via variable-range interaction for different values of $\alpha$ in one dimension (see Fig. \ref{fig:ggm_signal} in which the behavior of $\mathcal{G}$ with the valuation of time for different $\alpha$ is shown). For any  $\alpha$ in the weighted graph states, the  maximum value of GGM, $\mathcal{G}=0.5$, is achieved at $t=(2n-1)\pi \forall n\in \mathbb{N}$, since GGM in case of the evolving Hamiltonian with the nearest-neighbor interactions, i.e., $z=1$, is independent of $\alpha$ in Eq. (\ref{eq:ggm_open}) and $\cos \frac{t}{2r^\alpha} =0$ $\forall \alpha$ and $r=1$. To detect the transition between non-local and quasi-local or regimes, we use two identifiers, given by
\begin{equation}
    \bar{\mathcal{G}}_{2\pi}(\alpha)=\left.\frac{d\mathcal{G}(\alpha,t)}{dt}\right|_{t=2\pi};\quad \frac{d\mathcal{G}_{2\pi}}{d\alpha}(\alpha)=\frac{d}{d\alpha}\mathcal{G}(\alpha,t=2\pi).
\end{equation}
\begin{figure}[h]
    \centering
    \includegraphics[width=\linewidth]{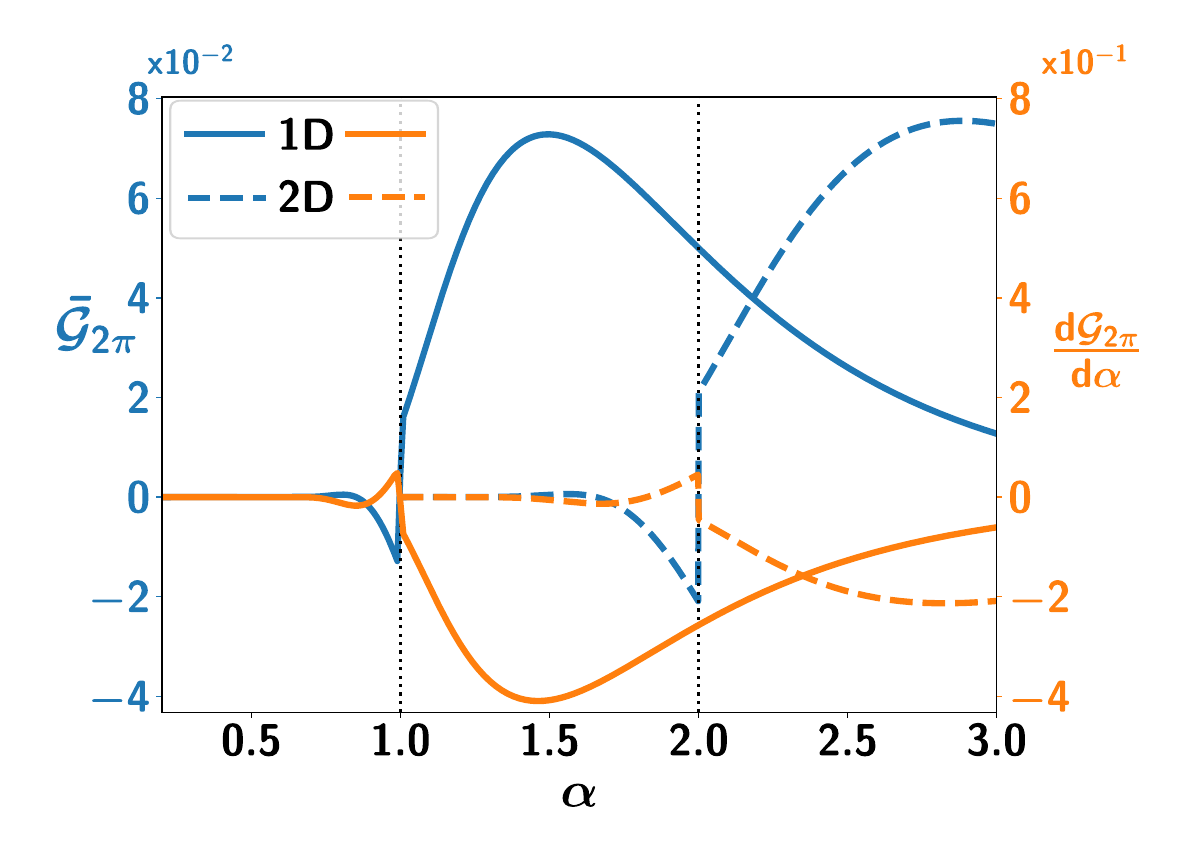}

    \caption{(Color online.) $\bar{\mathcal{G}}_{2\pi}$ (left vertical axis, blue/dark-shaded lines) and $\frac{d\mathcal{G}_{2\pi}}{d \alpha}$ (right vertical axis, orange/light-shaded) as a detector of non-local to quasi-local transition ($\alpha_{LR}^*$) against fall-off rate $\alpha$ (horizontal axis) for a $1$D chain (solid lines) and $2$D square lattice (dashed lines). While $\bar{\mathcal{G}}_{2\pi}$ is discontinuous at $\alpha_{LR}^*$, $\frac{d\mathcal{G}_{2\pi}}{d \alpha}$ is non-analytic, with the second derivative with respect to \(\alpha\) being discontinuous at $\alpha_{LR}^*$. The system sizes are  $N=5000$ for the $1$D chain and $40\times40$  for the $2$D square lattice with $z=N-1$. All axes are dimensionless. }
    \label{fig:detection}
\end{figure}
Let us first discuss the extreme points.  $\alpha=0$ corresponds to a uniform graph state where each vertex is connected to every other vertex with identical weights \cite{mahto_shaji_pra_2022}, and $\alpha\to\infty$ leads to a graph state with uniform weights involving the nearest-neighbor  vertices \cite{briegel01, rauss_pra2003}. It  becomes the cluster state  at $t=\pi,3\pi,5\pi,\ldots$ and GGM of the dynamical state obtained through the NN interacting Hamiltonian is oscillatory, with the time period of $2\pi$, i.e.,  $\mathcal{G}_{2\pi}=0$ and $\bar{\mathcal{G}}_{2\pi}=0$. For any finite $\alpha$, the weights of all connections contribute, such that these weights decrease polynomially with the increasing distance between vertices ($\phi_{i,i+r}=r^{-\alpha}$). For large $\alpha$, connections other than the nearest neighbors ($\phi_{i,i+1}=1)$ are extremely small and $\mathcal{G}_{2\pi}$ remains close to zero. With the decrease of $\alpha$, the long-range connections start to increase. For GGM at all times, typically $\mathcal{G}_{2\pi}$is nonvanishing since the absolute product of cosines at $t=2\pi$ in the second term of Eq. (\ref{eq:ggm_open}) keeps on decreasing from unity. On the other hand, at $\alpha=1$, the next-nearest neighbor make $\mathcal{G}_{2\pi}=0.5$ since $\cos\frac{2\pi}{4}=0$. When $\alpha$  goes below unity, weights of further distance makes the absolute cosine product almost near zero. This observation serves as a compelling incentive for employing GGM in the identification of the $\alpha^*_{LR}$ transition point. 

\subsubsection{Identifying transitions in falling rate in one dimension.}
First, notice that $\bar{\mathcal{G}}_{2\pi}$ vanishes until $\alpha \approx 0.7$ and it takes negative value upto $\alpha_{LR}^*=1$ while for $\alpha>\alpha_{LR}^*$, $\bar{\mathcal{G}}_{2\pi}$ abruptly changes to a positive value and starts increasing as shown in Fig. \ref{fig:detection}. Secondly, similar discontinuity at $\alpha_{LR}^*$ can also be observed from the quantity $\dv{\mathcal{G}(\alpha,t=2\pi)}{\alpha}=\dv{\mathcal{G}_{2\pi}}{\alpha}$. Both derivatives of GME content of the dynamical state, $\bar{\mathcal{G}}_{2\pi}(\alpha)$ and $\dv{\mathcal{G}_{2\pi}}{\alpha}$, can successfully identify the transition at $\alpha_{LR}^*=1$ from non-local to quasi-local regimes in one dimension by changing its characteristic instantly.

Let us determine whether the GME state is capable to detect another transition point present, $\alpha_{LR}^*$. To examine it, we first define the averaged GGM $\langle \mathcal{G}\rangle_T$ with averaging being performed over time, $t=T$ as 
\begin{eqnarray}
    \langle\mathcal{G}\rangle_T &=& \frac{\int_0^{T}\mathcal{G}(t)dt}{T}.
\end{eqnarray}
\begin{figure}[h]
    \centering
    \includegraphics[width=\linewidth]{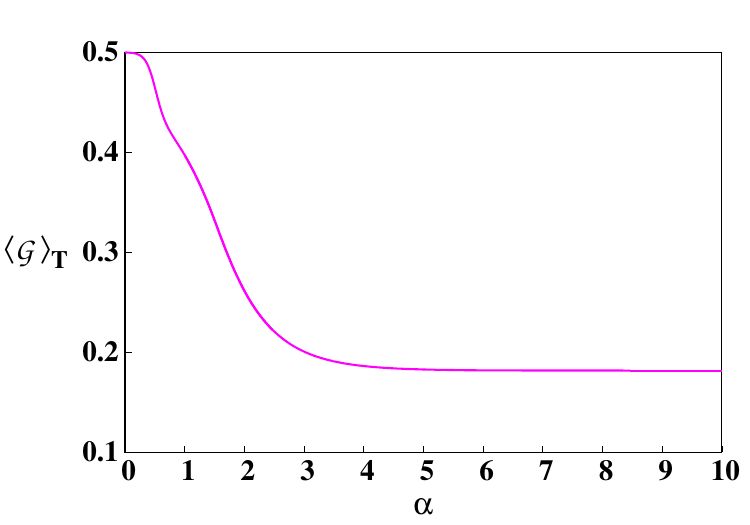}
    \caption{(Color online.) {\bf Nonlocal and quasilocal vs. local regimes.} The time-averaged value of GGM, $\langle\mathcal{G}\rangle_{T=3\pi}$ (vertical axis) upto time $T=3\pi$ against the interaction parameter, $\alpha$ (horizontal axis). Here the $N=10^6$-qubit state  is generated with coordination number \(z=N-1\)  . Starting from the maximum value of $\langle\mathcal{G}\rangle_{T=3\pi}=0.5$ at $\alpha=0$ corresponding to the end-to-end connected LR system with equal interaction strengths, it decreases drastically until the quasilocal region and finally saturates to $\langle\mathcal{G}\rangle_{T=3\pi}\approx 0.18$ for $\alpha \gtrsim 5$, i.e., in the local regime. Both the axes are dimensionless. }
    \label{fig:avg_ggm}
\end{figure}
Since $\mathcal{G}$ typically shows some repetitive behavior, with time $T$, the upper limit can be chosen according to such repetitions in $\mathcal{G}$. Extensive numerical simulations indicate that $T=3\pi$ is a good choice as it can capture all the features considered in this paper. For $\alpha\rightarrow \infty$ which corresponds to the nearest-neighbor case, we have $\langle\mathcal{G}\rangle_{T=3\pi}\equiv0.18$ while it reaches its maximum value $\langle\mathcal{G}\rangle_{T=3\pi}\equiv0.5$ for $\alpha=0$. In the intermediate $\alpha$, we observe contrasting behavior in $\langle\mathcal{G}\rangle_{T}$ when $0<\alpha<2$, i.e., the evolving Hamiltonian belongs to the long-range and quasi-local domains, $\langle\mathcal{G}\rangle$ sharply decreases with the increase of $\alpha$ as shown in Fig. \ref{fig:avg_ggm}, $\langle\mathcal{G}\rangle_{T}$ remains almost constant with $\alpha>>2$, i.e., when the dynamical state arises due to the short-range interaction. It indicates that the time-averaged GGM carries the signature of the transition point separating long-range and quasi-local models from the short-range model.

\begin{figure}[h]
    \centering
    \includegraphics[width=\linewidth]{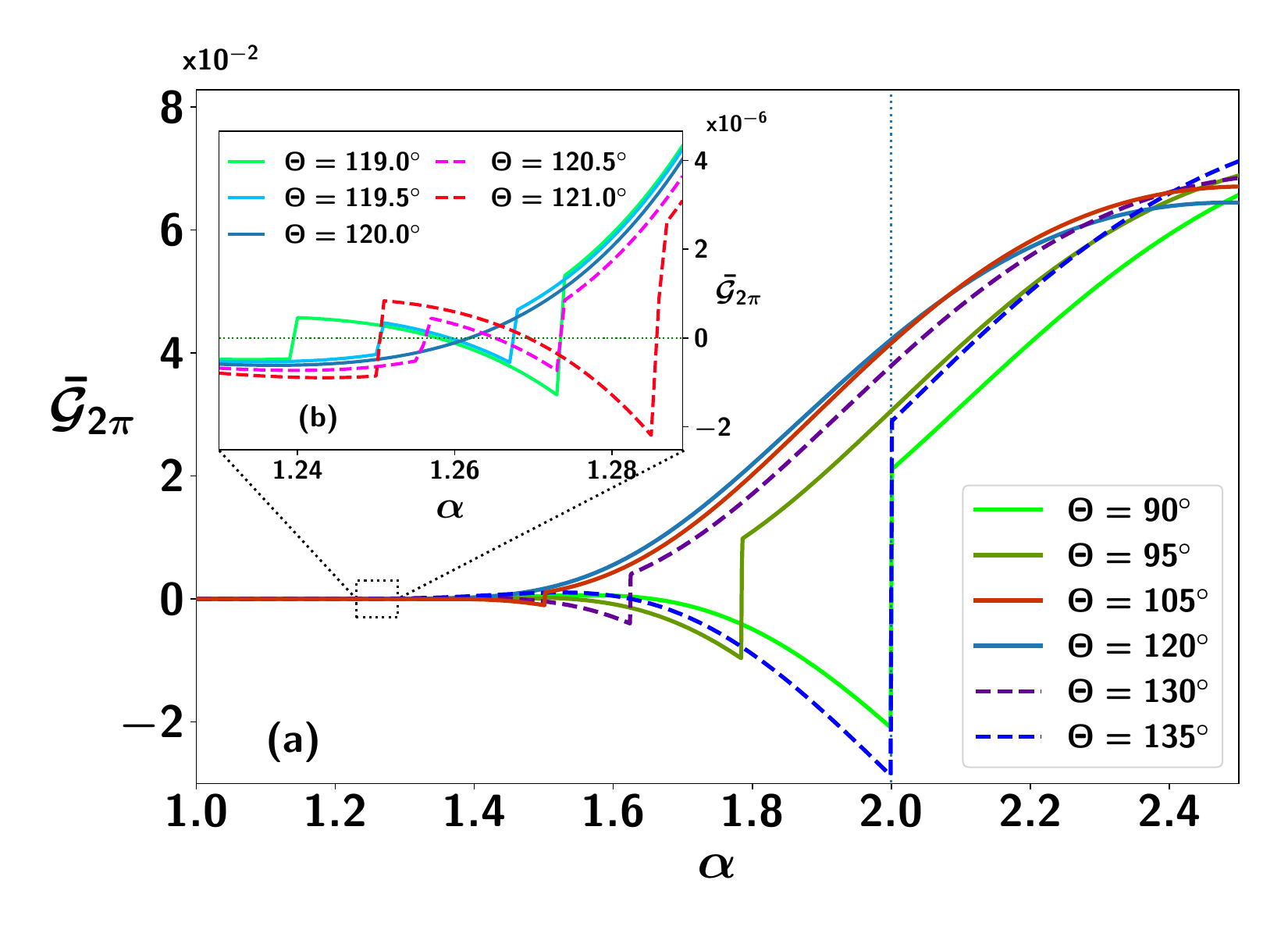}
    \caption{(Color online.) {\bf Detection of transition between nonlocal  and quasi-local orders. } $\bar{\mathcal{G}}_{2\pi}$ (ordinate) is shown at $\alpha_{LR}^*$ with respect to the fall-off rate $\alpha$ (abscissa) for different $2$D lattice geometries. (a) With increasing distortion angle, $\Theta$ from $90^\circ$ (the square lattice), $\alpha_{LR}^*$ (the point in which $\bar{\mathcal{G}}_{2\pi}$ shows discontinuity) decreases with \(\alpha\) till $\Theta=120^\circ$ (the honeycomb lattice). When the distortion angle is increased further, $\alpha_{LR}^*$ is now increasing with $\Theta$, and $\alpha_{LR}^*=2$ is obtained for $\Theta=135^\circ$. (b) Although $\bar{\mathcal{G}}_{2\pi}$ is continuous for the honeycomb lattice, it is still discontinuous when $\Theta$ is close to $120^\circ$, which can be used to obtain the transition point for the honeycomb lattice in the asymptotic limit. The system is $N=40\times40$ in the $2$D square lattice with all-to-all connection $(z=N-1)$. All axes are dimensionless.}
    \label{fig:detection_theta}
\end{figure}

\subsection{Determining transition in higher dimensional lattices via GGM.} 
Let us extend a similar analysis of the GME state acquired by the weighted graph state obtained in two-dimensional lattices. As illustrated in Fig. \ref{fig:schematic3}, we begin with a square lattice which is then deformed to other lattice structures with the deformation being given by an angle $\Theta$. Both $\bar{\mathcal{G}}_{2\pi}$ and $\frac{d\mathcal{G}_{2\pi}}{d\alpha}$ exhibit discontinuity at $\alpha_{LR}^*=2$, thereby showing their capability to detect the transition point present in a square lattice (see Fig. \ref{fig:detection_theta}).

Let us concentrate on the other deformed lattices and their transition points which, intuitively, depend on $\Theta$. While a similar behavior of GGM with time is observed in the weighted graph state on these distorted lattices, the transition point of fall-off rate $\alpha$ changes depending on $\Theta$ as evident from Fig. \ref{fig:detection_theta}. To make the study quantitatively, we introduce a quantity
\begin{eqnarray}
    \Delta\bar{ \mathcal{G}}_{2\pi}^{\Theta}=\bar{\mathcal{G}}(\alpha^*_{LR}+\delta,t)-\bar{\mathcal{G}}(\alpha^*_{LR}-\delta,t),
\end{eqnarray}
characterising the jump observed at $\alpha_{LR}^*$ for a small increment in \(\alpha_{LR}^*\), denoted as \(\delta\) (which we have taken as \(0.001\) for our investigations). As the distortion $\Theta$ is increased (from $90^\circ$, in the case of the square lattice), $\alpha_{LR}^*$ decreases from $2$ with the discontinuity ($\Delta\mathcal{\bar{G}}_{2\pi}$) being decreased. The trends continue with $\Theta$ till the lattice gets the honeycomb structure, i.e., $\Theta=120^\circ$, where $\mathcal{\bar{G}}_{2\pi}$ is continuous ($\Delta\mathcal{\bar{G}}_{2\pi}=0$) as shown in the inset, Fig. \ref{fig:detection_theta}(b). 
\begin{figure}[h]
    \centering
    \includegraphics[width=\linewidth]{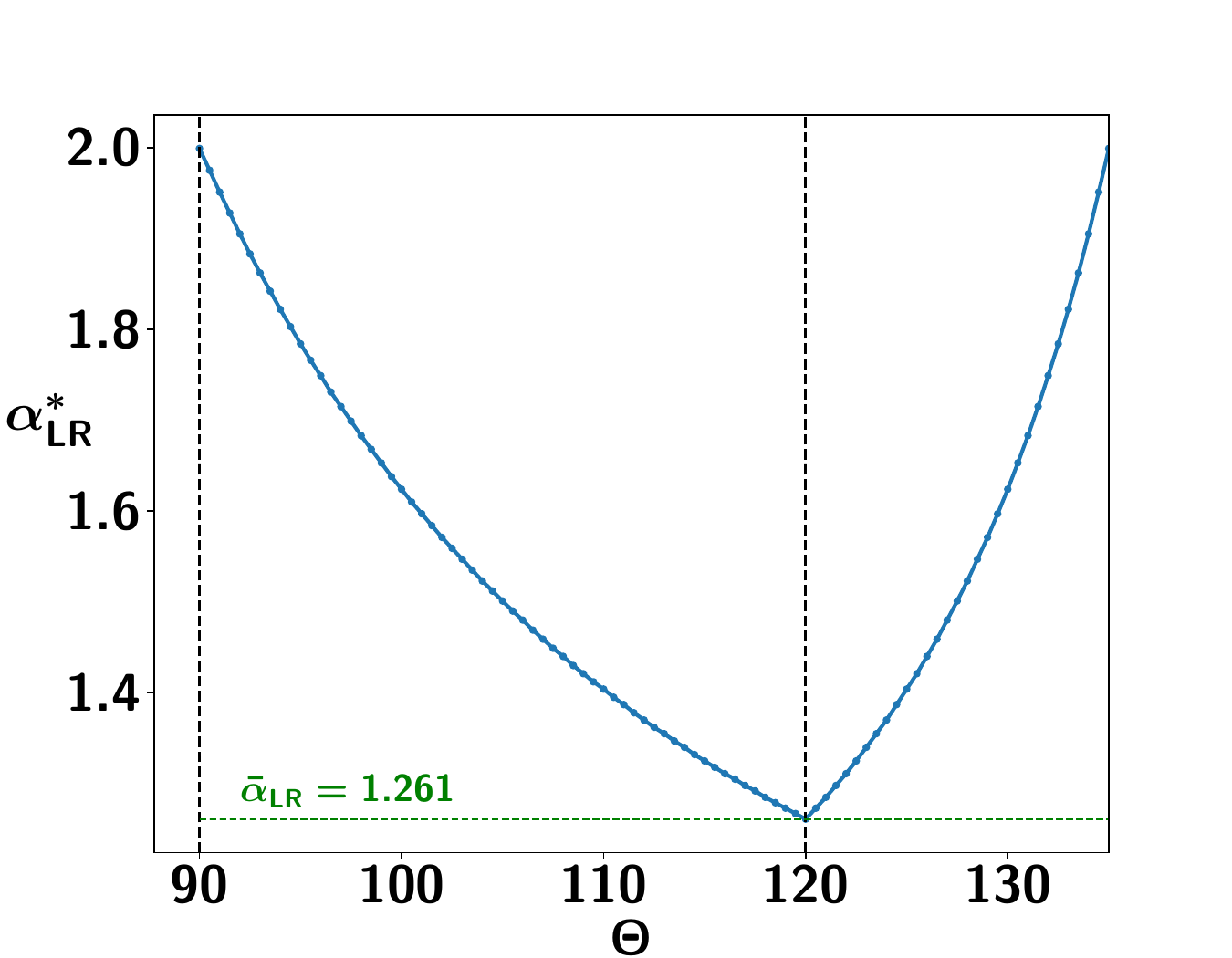}
    \caption{(Color online.) {\bf Identifying transitions from non-local to local regime in $2$D lattices.} Dependence of the non-local to quasi-local transition point, $\alpha_{LR}^*$ (vertical axis) on the distortion angle $\Theta$ (horizontal axis), is shown. The transition point in the honeycomb lattice ($\Theta=120^\circ$) is obtained asymptotically, where it has the minimum value. Notice that it will be interesting  to confirm that the predictions of transitions made by the GME content of the dynamical states  in the deformed square lattices match with the transition points obtained through the correlation length of the ground state.  All axes are dimensionless.}
    \label{fig:alpha_theta}
\end{figure}
Specifically, for $120^{\circ}>\Theta_1>\Theta_2>90^{\circ}$, we find $\Delta\mathcal{\bar{G}}_{2\pi}^{\Theta_1}<\Delta\mathcal{\bar{G}}_{2\pi}^{\Theta_2}$. The limiting case with $\Theta\to120^\circ$ corresponds to the $\alpha$ value ($\bar{\alpha}_{LR}$) where $\mathcal{\bar{G}}_{2\pi}=0$ continuously and changes its sign at $\alpha=\bar{\alpha}_{LR}$. Further, both $\alpha_{LR}^*$ and $\Delta\mathcal{\bar{G}}_{2\pi}$ start increasing, when $\Theta>120^{\circ}$. The transition point, $\alpha_{LR}^*=2$, also emerges for the distorted $2$D lattice with $\Theta=135^\circ$, but with higher $\Delta\mathcal{\bar{G}}_{2\pi}$ than of $\Theta=90^\circ$. The inset (Fig. \ref{fig:detection_theta}(b)) shows $\Delta\mathcal{\bar{G}}_{2\pi}$ is of the order of $10^{-6}$, which abruptly vanishes for $\Theta=120^\circ$. Notice that  another discontinuity in $\mathcal{\bar{G}}_{2\pi}$ is observed near $\Theta=120^\circ$ which decreases as $\Theta$ is far away from $120^\circ$ and the reason behind such observation is not clear from our paper. 

\begin{figure*}
\centering
    \includegraphics[width=\linewidth]{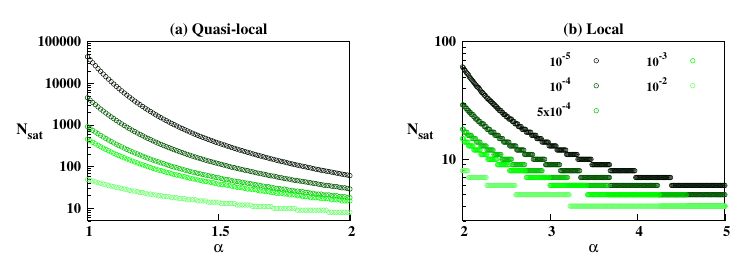}
    \caption{(Color online.) Saturation value of the number of qubits $N_{sat}$ (ordinate in logarithmic scale) against interaction parameter $\alpha$ (abscissa) for different values of error $\epsilon$ (mentioned in legends) in the (a) quasi-local and (b) local regions. Here \(z= N-1\). Both the axes are dimensionless.}
    \label{fig:Nsat-semilocal}
\end{figure*}

Summarizing the entire analysis establishes the dependence of $\alpha_{LR}^*$ on $\Theta$. In particular, $\alpha_{LR}^*$ decreases when $90^{\circ}<\Theta<120^{\circ}$  while it again increases with $120^{\circ}<\Theta<135^{\circ}$ (see Fig. \ref{fig:alpha_theta}). Notice that $\alpha_{LR}^*$ shows a sudden change at $\Theta=120^{\circ}$, although $\mathcal{\bar{G}}_{2\pi}$ is smooth for the honeycomb lattice, with the limiting value being $\bar{\alpha}_{LR}^*=\alpha_{LR}^*(\Theta\to120^\circ)=1.261\pm0.001$. 
It is worth noticing that when $135^{\circ}<\Theta<150^{\circ}$, \(\alpha_{LR}^*>2\) although beyond \(\Theta>150^{\circ}\), the unit chosen for NN sites is violated.   
Furthermore, studying the pattern of the GME state generated with variable-range interaction turns out to be an efficient method for detecting transition points in the fall-off rate. 

\subsection{Replicating the time-averaged GGM of an entire lattice with a smaller system size}
\label{sec:nsat}
We now concentrate on the reduction of the resource or the depth of the circuit. Specifically we wish to investigate the presence of a critical value concerning the total number of qubits at which the time-averaged multiparty entanglement saturates despite the further augmentation of qubit count. It can be achieved by minimizing the total number of qubits for all-connected weighted graph states with large number of qubits such that its time-averaged GGM  is emulated. To address this question mathematically, we define a saturation value in the total number of qubits, denoted as 
\begin{eqnarray}
    \nonumber \Delta\langle\mathcal{G}\rangle\equiv \text{min}_{N}| \langle\mathcal{G}(N+1,\alpha)\rangle_T - \langle\mathcal{G}(N,\alpha)\rangle_T |<\epsilon, \\
\end{eqnarray}
for an infinitesimal number $\epsilon$, and the corresponding $N$ represents $N_{sat}$. Here $\epsilon$ is fixed from the accuracy reached in the computation.

Let us illustrate that $N_{sat}$ also carries the signature of the transition point in $\alpha$. For examination, we fix  $\epsilon =10^{-2},10^{-3},5\times 10^{-3},10^{-4}$ with $T=3\pi$.
For $1<\alpha\lesssim 2$, the variation of $N_{sat}$ is extremely rapid, which is not the case for $\alpha>>2$. For example, with $\epsilon=10^{-4}$ in the LR model with $\alpha=1$, we obtain $N_{sat}=4521$, which decreases drastically to $N_{sat}=29$ for $\alpha = 2$. It declines further with the increase of $\alpha$ and $N_{sat}|_{\alpha = 3}=9$ and  $N_{sat}|_{\alpha = 5}=5$ in the local regime. Therefore, $N_{sat}$ becomes single-digit when only short-range interactions are involved in the evolving Hamiltonian while it is moderate in the quasi-local regime and possesses a very high value in the LR system.

With the decrease of $\epsilon$, the saturation value $N_{sat}$ grows rapidly in the region $1\lesssim\alpha\lesssim 2$, although beyond this region the rate of increase is much lower, i.e., the value of $N_{sat}$ remains unaltered as shown in Fig. \ref{fig:Nsat-semilocal} for $\alpha\gtrsim2$. Let us ask the question, ``What is the minimum number of qubits required so that the time-averaged GGM $\langle\mathcal{G}\rangle_{T=3\pi}$  saturates for a given $\alpha$?" As shown in Fig.  \ref{fig:avg_ggm}, with the model having $N=10^6$ sites, $\langle\mathcal{G}\rangle_{T=3\pi}=0.331971$ at $\alpha=1.5$. By careful analysis, we find that with $\epsilon=10^{-4}$ and with $ N_{sat}=117$, which is the moderate number of qubits, achievable even with current technologies, we can obtain $\langle\mathcal{G}\rangle_{T=3\pi}=0.331961$. However, in the presence of long-range interactions,  the error has a drastic effect on $\Delta\langle\mathcal{G}\rangle$ and it is nearly impossible to find $N_{sat}<<N$ in this domain.

\section{Mimicking the GGM from longrange with the short-range model}
\label{sec:mimick}
The emergence of exquisite traits in LR models that are often missing in few-body interactive systems has attracted lots of interest. Apart from the physical systems like trapped ions in which the long-range model arises naturally \cite{HAFFNER2008155}, there exist other physical systems including superconducting circuits where one requires two-qubit gates to generate interactions between distant sites \cite{ying23, PRXQuantumsuper}. With the increase of range of interaction, the number of two-qubit gates increases and hence the decoherence effects become prominent. Without compromising the production of the GME state, our objective is to decrease the range of interactions, i.e., $z$, which can mimic the entanglement properties of LR.
\begin{figure}[h]
    \centering
    \includegraphics[width=\linewidth]{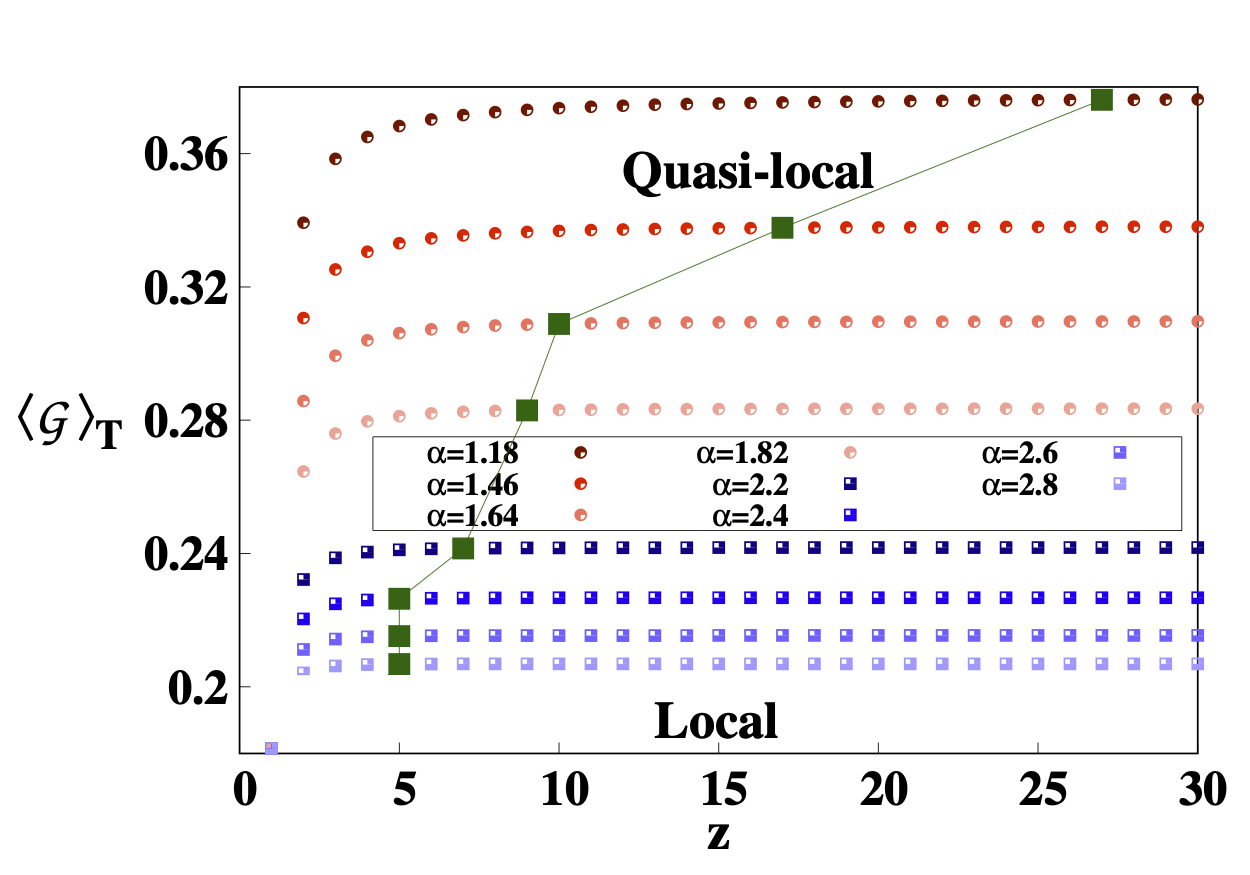}
    \caption{(Color online.) The time-averaged value of GGM, $\langle\mathcal{G}\rangle$ (\(y\)-axis) upto time $T=3\pi$ for $N=1200$ against coordination number, $z$ (x axis). Different patterns (and colors) indicate quasi-local and local regimes while different shades of a particular color correspond to different values of the interaction parameter $\alpha$. The dark green square boxes suggest the critical value $z=z_c$ as in Eq. (\ref{eq_zc}). Here \(\epsilon =0.0001\). Notice that the change of $z_c$ with the change of $z$ can predict the existence of $\alpha^*_{SR}$. Both the axes are dimensionless.}
    \label{fig:Z_critical}
\end{figure}
Specifically, for a fixed value of $N$,  we ask "what is the minimum coordination number required to produce GGM that can be obtained with the fully connected model?" We study the trade off between genuine multiparty entanglement content of a fully connected model and the same created by evolving the Hamiltonian with the finite range of interactions. 

Let us define the minimum coordination number. For a fixed $N$, and $\alpha$, if the difference of the time-averaged GGM with the variation of $z$ remains constant, we call the minimum $z$ in which $\langle \mathcal{G}\rangle$ saturates as $z_c$. Mathematically, we compute 
\begin{eqnarray}
    \Tilde{\Delta}\langle\mathcal{G}\rangle_{z_c} \equiv \min_{z}| \langle\mathcal{G}(z=N-1)\rangle_T - \langle\mathcal{G}(z)\rangle_T|<\epsilon.
    \label{eq_zc}
\end{eqnarray}
with $\epsilon$ being the infinitesimal number which leads to $z_c$. Like in the previous case for all values of $\alpha$, we set $T=3\pi$, $\epsilon\approx10^-3$ in the quasi-local and local regions. In the quasi-local and local regions, we observe that it is indeed possible to find $z_c$, i.e., there exist $z_c$ for a fixed $\alpha$ and $N$, above which $\langle \mathcal{G}\rangle_{T}$ is constant. For example, with $\alpha=1.82$ and $N=120$, the time-averaged GGM obtained with the LR model ($z=N-1$) can be attained with $z_c=41$. For different $\alpha<2$ values, such $z_c$ emerges for a fixed $N$ from $\langle \mathcal{G}\rangle_{3\pi}$ as depicted in Fig. \ref{fig:Z_critical}. Furthermore, the variation of $z_c$ with $\alpha$ can determine the critical point, $\alpha_{SR}^*$. Specifically, in the local regime, $z_c$ remains almost unaltered with $\alpha$ while its increase is drastic when $1\lesssim\alpha\lesssim 2$ (comparing Figs. \ref{fig:Nsat-semilocal}(a) and \ref{fig:Nsat-semilocal}(b)).

\section{ Strategy of disconnecting subgraphs in a WGS}
\label{sec:measurement_effect}
Upto now, we have investigated the entanglement pattern in the WGS and connect it with the transition points present in the LR Ising Hamiltonian. We now move our attention to two different queries.

(1) Does the reduced state become a WGS after removing the few redundant qubits through measurement from the lattice which are not useful for a certain task?

(2) Like graph states, does the algebra of Pauli measurements required for MBQC exist for WGS? 

We clearly obtain an affirmative answer for the first query while our paper provides an indication to develop such algebra for the WGS, which can be  a possible future direction of research.

Motivated by these questions, our objective, in this section, is to design a measurement strategy such that the remaining unmeasured parties share the WGS. Such scheme diminishes the qubit count from $N$ to $(N-m)$ while preserving an $(N-m)$-qubit WGS, all the while maintaining the original inter-qubit separation. Fig. \ref{fig:local_measurement} illustrates the schematic diagram  to achieve the same for $N=5$ and $m=3$ on which the local measurement operators act. Notice that the strategy presented below is valid for both 1D and 2D lattices.

\subsection*{ Towards qubit number reduction in a WGS.} Let us perform a single-qubit projective measurement $\{M_k^{j}|j=1,2\}$ on the qubit $k$ of the $N$-qubit WGS. Therefore, when the outcome $j$ occurs, the normalised  output state after tracing out the $k$th qubit can be presented as 
\begin{eqnarray}
    \nonumber &&\ket{\Psi(t)}_{N-1\{12\ldots {(k-1)} {(k+1)}\ldots N \}}^{j_k} \\ &&=\frac{M_k^{j}\ket{\Psi(t)}_{N\{1\ldots N\}}}{\sqrt{_{N\{1\ldots N\}}\bra{\Psi(t)}M_k^{j}\ket{\Psi(t)}_{N\{1\ldots N\}}}}.
\end{eqnarray}

\begin{theorem}
    Single qubit measurements in $\sigma_z$ basis on $m$ number of qubits situated in arbitrary position on an all-to-all connected $N$-qubit weighted graph state provide a local unitarily equivalent $(N-m)$-qubit weighted graph state with modified weights where initial qubit position remains fixed.
    \label{th:2}
\end{theorem}

\begin{proof}
Let us rewrite an $N$-qubit weighted graph state in Eq. (\ref{eq:wgs}), following the power law as
\begin{eqnarray}
   \nonumber &&\ket{G(t)}_{N\{1,2,\ldots,k,\ldots,N\}} \\
   &&=\frac{1}{2^{N/2}}\sum_{\eta=0}^{2^N-1}\exp\Big(-i\sum_{i=1}^{N-1}\sum_{j>i}^{N} g_{i j}(t) a_i a_j\Big)\ket{\eta},
\end{eqnarray}
which consists of all possible states in the computational basis of $N$ qubits with $\ket{\eta}$ being the decimal equivalent of binary values. Performing measurement in  the $\sigma_z$ basis on the $k$th qubit and selecting the outcome $M^0_k=\ket{0}_k\bra{0}$, the output state takes the form
\begin{eqnarray}
    \nonumber &&\ket{G(t)}_{N\{1,2,\ldots,k,\ldots,N\}}^{0_k} \\ \nonumber
    &&=\frac{1}{2^{N/2}}\sum_{\substack{r=2\mathbb{Z}\\r<2^k}}\sum_{\eta=r2^{N-k}}^{(1+r)2^{N-k}-1}\exp\Big(-i\sum_{i=1}^{N-1}\sum_{j>i}^{N} g_{i j}(t) a_i a_j \Big)\ket{\eta}.\\
\end{eqnarray}
Tracing out the $k$th qubit reduces the binary value of each basis to $(\eta'-r2^{N-k})\forall r$ by which we write the corresponding state
\begin{eqnarray}
    \nonumber &&\ket{G(t)}_{(N-1)\{1,2,\ldots,k-1,k+1,\ldots,N\}}^{0_k} \\ \nonumber
    &&=\frac{1}{2^{(N-1)/2}}\sum_{\substack{r=2\mathbb{Z}\\r<2^k}}\sum_{\eta=r2^{N-k}}^{(1+r)2^{N-k}-1}\exp\Big(-i\sum_{i=1}^{N-1}\sum_{j>i}^{N} g_{i j}(t) a_i a_j \Big)\\ \nonumber &&\ket{\eta-r2^{N-k}}.\\
    \label{eq:eta-r2nk}
\end{eqnarray}
\begin{figure}[h]
    \centering
    \includegraphics[width=\linewidth]{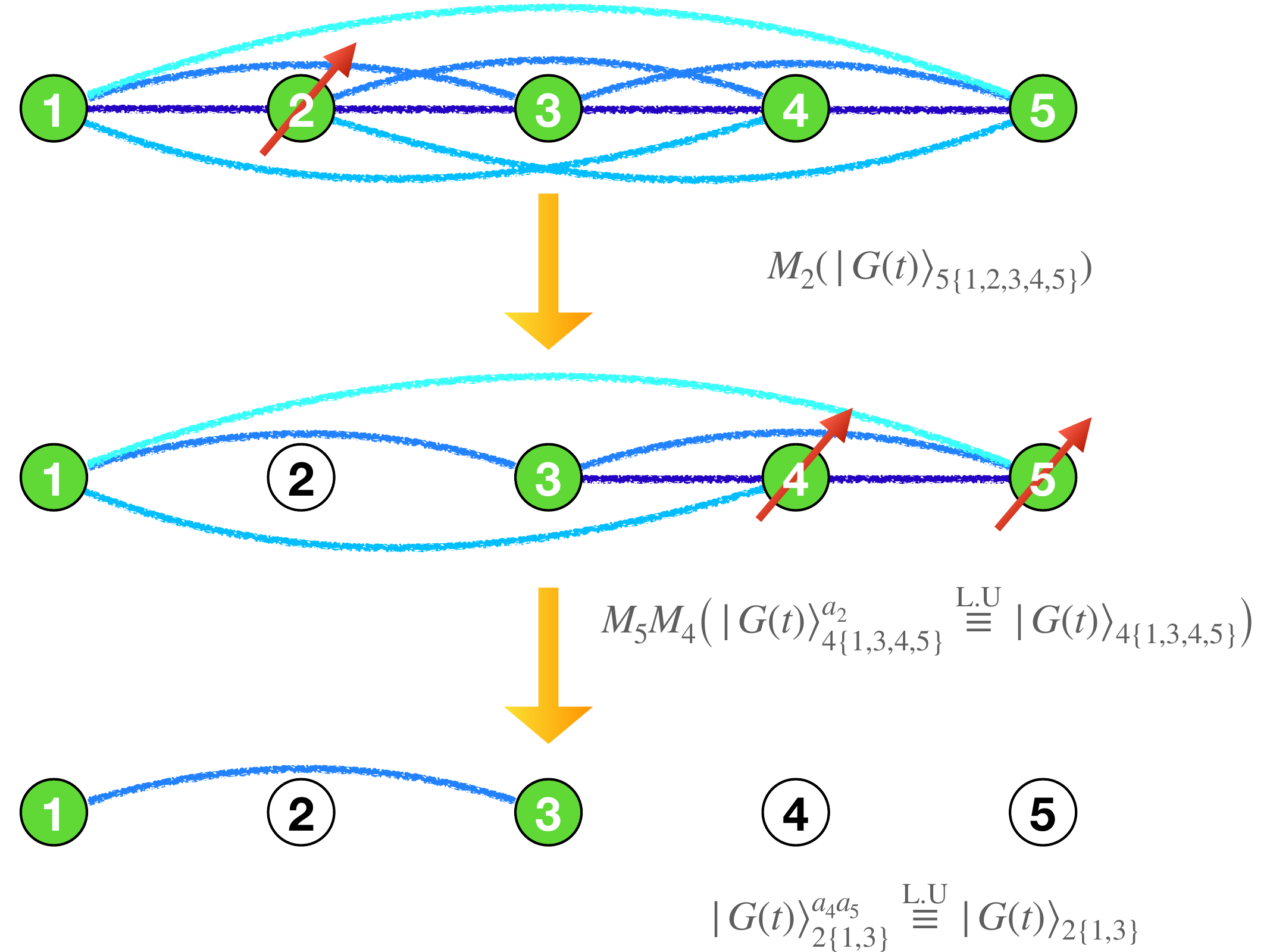}
    \caption{{\bf Decoupling scheme of a required circuit via local measurements.} Schematic diagram of the protocol for reducing the number of qubits in WGS through local measurements along the $z$-axis. For illustration, let us consider a five-party WGS, $\ket{G(t)}_{5\{1,2,3,4,5\}}$. Performing the first measurement $M_2$ on the second  qubit, the four-party state becomes (upto normalization) $M_2(\ket{G(t)}_{5\{1,2,3,4,5\}})=\ket{G(t)}^{a_2}_{4\{1,3,4,5\}}$ which is local unitarily equivalent to $\text{LU}\ket{G(t)}_{4\{1,3,4,5\}}$. Again performing successive measurements in qubits $4$ and $5$,  the two-party output state takes the form $M_5M_4(\ket{G(t)}_{4\{1,3,4,5\}})=\ket{G(t)}^{a_5a_4}_{2\{1,3\}}\equiv\text{LU}\ket{G(t)}_{2\{1,3\}}$. }
    \label{fig:local_measurement}
\end{figure}
Eventually for each possible value of $\eta$ and $r<2^k$, all possible terms in Eq. (\ref{eq:eta-r2nk}) form a basis of the $(N-1)$-qubit WGS state. Therefore, we can write 
\begin{eqnarray}
    \nonumber &&\ket{G(t)}_{(N-1)\{1,2,\ldots,k-1,k+1,\ldots,N\}}^{0_k} \\ \nonumber&&=\frac{1}{2^{(N-1)/2}}\sum_{\zeta=0}^{2^{N-1}-1}\exp\Big(-i\sum_{\substack{i=1\\i\neq k}}^{N-1}\sum_{\substack{j>i\\j\neq k}}^{N} g_{i j}(t) a_i a_j \Big)\ket{\zeta}\\ \nonumber &&=\prod_{\substack{i=1\\i\neq k}}^{N-1} \prod_{\substack{j>i\\j\neq k}}^{N}U_{ij} \ket{+}_1\ket{+}_2\otimes\ldots\otimes\ket{+}_{k-1}\ket{+}_{k+1}\otimes\ldots\otimes\ket{+}_N \\ \nonumber &&=\ket{{G}(t)}_{(N-1)\{1,2,\ldots,k-1,k+1,\ldots,N\}},
\end{eqnarray}
where $\ket{\zeta}\equiv\ket{a_1 a_2 \ldots a_{k-1} a_{k+1}\ldots a_N}$. 
Similarly, measuring and tracing out site $l$ on the $(N-1)$-qubit WGS state in the computational basis leads to the $(N-2)$-qubit WGS state when the outcome is $M_l^0$ and so on. After $m$ number of such measurements, we obtain the $(N-m)$-qubit WGS state, when $M^0$ clicks in all the $m$ sites. Specifically, measuring on $k_1,\ldots,k_m$ sites, the resulting state is the $(N-m)$-qubit WGS, denoted as $\ket{{G}(t)}^{0_{k_1}\ldots 0_{k_m}}_{(N-m)\{k\}_{k=1;k\neq k_1,\ldots, k_m }^N}\equiv\ket{{G}(t)}_{(N-m)\{k\}_{k=1;k\neq k_1,\ldots, k_m }^N}$. The single-site density matrix corresponds to the site $l$ as
\begin{eqnarray}
   \rho_l^0(k_1,\ldots,k_m) &=&\begin{pmatrix}
    \frac{1}{2} & x\\
    x^* & \frac{1}{2}
  \end{pmatrix} \equiv \rho_l,
\end{eqnarray}
where
\begin{eqnarray}
    x=\frac{1}{2^{N-m}}\prod_{\substack{j=1 \\ j\neq l, k_1,\ldots,k_m}}^N(1+e^{i \phi_{lj}}),
\end{eqnarray}
and $\rho_l^0(k_1,\ldots,k_m)$ denotes the $l$th side density matrix with the measurement outcome on $k_1,\ldots,k_m$ qubits being $M^0$.

Let us consider the situation when $M^1$ clicks. The single-site density matrix changes to $\rho_l^1(k_1,\ldots,k_m)$ with different off-diagonal entries modified by a phase factor as
    \begin{eqnarray}
    \rho_l^1(k_1,\ldots,k_m)=\begin{pmatrix}
    \frac{1}{2} & y\\
    y^* & \frac{1}{2}
  \end{pmatrix},
\end{eqnarray}
where
\begin{eqnarray}
    y=\frac{1}{2^{N-m}}e^{i\sum_{r=1}^m \phi_{lk_r}}\prod_{\substack{j=1 \\ j\neq l, k_1,\ldots,k_m}}^N(1+e^{i \phi_{lj}}).
\end{eqnarray}
The eigenvalues of $\rho_l^1(k_1,\ldots,k_m)$ and $\rho_l^0(k_1,\ldots,k_m)\equiv\rho_l$ are the same due to the fact that off-diagonal entries differ only by a phase factor. Therefore, the output state obtained with the outcome $M^1$ is local unitarily connected with the state having $M^0$ outcome. Here the local unitary at each site is given by 
\begin{eqnarray}
    \mathbb{U}_l=\begin{pmatrix}
    e^{i\sum_{r=1}^m \phi_{lk_r}} & 0 \\
    0 & 1
  \end{pmatrix},
\end{eqnarray}
i.e., 
\begin{eqnarray}
   \nonumber &&\ket{G(t)}^{0_{k_1},\ldots,0_{k_m}}_{(N-m)\{k\}_{k=1;k\neq k_1,\ldots, k_m }^N}\\&&= \bigotimes_{l\in k}\mathbb{U}_l\ket{G(t)}^{1_{k_1},\ldots,1_{k_m}}_{(N-m)\{k\}_{k=1;k\neq k_1,\ldots, k_m }^N}.
\end{eqnarray} 
Notice that $\mathbb{U}_l \in SU(2)$, in general, although for some particular values of $\sum_{r=1}^m \phi_{lk_r}=2n\pi,(4n+3)\frac{\pi}{2},(4n+1)\frac{\pi}{2}$, the local unitary $\mathbb{U}_l$ becomes Clifford unitary \cite{nielsen_chuang_2010} as it preserves the single-qubit Pauli group. Finally, it is clear that for total $m$ number of measurements on arbitrary  qubits, $k_1,\ldots,k_m$, we have to apply the local unitary on each qubit depending on the outcome at every step of the measurement  except $k_1,\ldots,k_m$, i.e., 

\begin{eqnarray}
    &&\ket{G(t)}_{(N-m)\{k\}_{k=1;k\neq k_1,\ldots, k_m }^N}\\ \nonumber&&= \prod_{j=1}^{m}\bigg(\mkern-18mu\bigotimes_{\substack{l_j=1 \\ l_j\neq k_1,\ldots, k_m}}^N \mkern-18mu \big( (1-a_{k_j}) \mathbb{I_{l_j}}+ a_{k_j}\mathbb{U}_{l_j})\bigg)\ket{G(t)}^{a_{k_1},\ldots,a_{k_m}}_{(N-m)\{k\}_{k=1;k\neq k_1,\ldots, k_m }^N}.
\end{eqnarray}
\end{proof}
Let us illustrate Theorem \ref{th:2} with an example. For $N=3$, after measuring and tracing out the second qubit, the output state becomes
\begin{eqnarray}
    \ket{G(t)}_{2\{1 3\}}^{0_2} =\frac{1}{2}(\ket{00}+\ket{01}+\ket{10}+e^{-i g_{1 3}}\ket{11}),
    \label{eq:3qubit_0}
\end{eqnarray}
when $\ket{0}$ clicks, which  is equivalent to $\ket{G(t)}_{2\{1 3\}} = U_{1 3}\ket{+}_{1}\ket{+}_{3}$ where the qubits at position $1$ and $3$ are connected. The output state of the outcome, $M^1=\ket{1}\bra{1}$, can be written as
\begin{eqnarray}
    \nonumber \ket{G(t)}_{2\{13\}}^{1_2}&=&\frac{1}{2}(\ket{00}+e^{-ig_{23}}\ket{01}+e^{-ig_{12}}\ket{10}\\ &+&e^{-i(g_{12}+g_{13}+g_{23})}\ket{11}),
\end{eqnarray}
which is equivalent to $\ket{G(t)}_{2\{13\}}$ up to the local unitary, $\mathbb{U}_1\otimes \mathbb{U}_3 = \text{diag}(e^{i\phi_{21}},1)_1\otimes \text{diag}(e^{i\phi_{23}},1)_3   $. Similarly, starting from $\ket{G(t)}_{5\{12345\}}$, we can generate $\ket{G(t)}_{4\{1345\}}$, $\ket{G(t)}_{3\{135\}}$, and $\ket{G(t)}_{2\{13\}}$ successively by measuring $\sigma_z$ on the second, fourth, and fifth qubits respectively as shown in Fig. \ref{fig:local_measurement}.

Furthermore, the rules of $\sigma_y$ and $\sigma_z$ measurements on the WGS, which turn out to be significantly complex and non-trivial compared to the graph states, can be an exciting research direction towards implementing MBQC.

\section{Conclusion}
	\label{sec:conclu}
 
MBQC is a promising candidate to implement quantum circuits in laboratories and hence it is crucial to identify the ingredient, known as the cluster state, which belongs to the set of stabilizer states required for performing MBQC perfectly. They are typically produced by using the nearest-neighbor Ising Hamiltonian from a product state. Instead of NN interactions, if the product state is evolved with variable-range interactions, weighted graph states possessing exquisite characteristics can be produced, exhibiting some features of the evolving Hamiltonian that NN interactions cannot reveal. We presented the exact expression of GGM as a function of fall-off rate, range of interactions, and time for any system size, in both one and two dimensions with varied geometries.
It is worth noting that producing GGM for an arbitrary number of sites is computationally demanding since it necessitates calculating the maximum eigenvalues of all possible reduced density matrices and hence with the increase of system size, the computational complexity grows exponentially.
 In this case, we proved that the maximum eigenvalue obtained only from the single-site reduced density matrix contributes to the GGM.
 We demonstrated that the time derivative of GGM in one and two dimensions, including lattices formed after deformation from square one, can identify transitions in fall-off rates from long-range to quasi-local regions. On the other hand, we found that the saturation of time-averaged GGM  with system size as well as fall-off rates can be employed to predict transitions from quasi-local to local regimes in the evolving Hamiltonian. Furthermore, we observed that, for a certain system size, GGM generated from a finite-range Hamiltonian resembles the behavior of GGM in the long-range models under the quasi-local and local regimes.

    We also investigated the effect of measurement along a certain direction on the weighted graph state. We showed that the resultant weighted graph state is local unitarily equivalent to another weighted graph state with fewer qubits, and obtained the relationship between the weights of the premeasured and the postmeasured states. 
   The present findings demonstrate the potential of the WGS as a resource for quantum communication and computation tasks. 
    

 \acknowledgements
	The authors acknowledge the support from  the Interdisciplinary Cyber Physical Systems (ICPS) program of the Department of Science and Technology (DST), India, Grant No.: DST/ICPS/QuST/Theme- 1/2019/23.  We  acknowledge the use of the cluster computing facility at the Harish-Chandra Research Institute.

\appendix
\section{Calculation of reduced density matrices of a WGS from the PEPS description}
\label{sec:WGS_PEPS}
Going to the PEPS description \cite{verstraete_pra_2004, dur_hartmann_prl_2005, hartmann_briegel_2007} of $\ket{\Psi^\prime}$ in Eq. (\ref{eq:psi_reduced}) based on $G^\prime$, each qubit in $B$ is replaced by $n$ virtual qubits (one for each site in $A$), and similarly, each qubit in $A$ is replaced by $N-n$ virtual qubits (one for each site in $B$). The  $\ket{\Psi^\prime}$ in terms of the virtual qubits can be represented as
\begin{equation}
    \ket{\bar{\Psi}^\prime}=\bigotimes_{k\in A,l\in B}U_{k,l}\ket{+}_{\bar{k}_l}\ket{+}_{\bar{l}_k}.
\end{equation}

\subsection{Action of projection on virtual qubits in $B$}
For simplicity, consider $N=3$ with $A=\{1,2\}$ and $B=\{3\}$. Denoting the virtual qubits of $1$ as $\bar{1}_3$, of $2$ as $\bar{2}_3$, and of $3$ as $\bar{3}_1$ and $\bar{3}_2$, in the PEPS form, the state is
\begin{eqnarray*}
    \ket{\bar{\Psi}^\prime}_{N=3} &=&\left(U_{13}\ket{+}_{\bar{1}_3}\ket{+}_{\bar{3}_1}\right) \left(U_{23}\ket{+}_{\bar{2}_3}\ket{+}_{\bar{3}_2}\right)\\
    &=&\ket{+}_{\bar{1}_3}\ket{+}_{\bar{2}_3}\ket{0}_{\bar{3}_1}\ket{0}_{\bar{3}_2} + \ket{\zeta_1}_{\bar{1}_3}\ket{\zeta_2}_{\bar{2}_3}\ket{1}_{\bar{3}_1}\ket{1}_{\bar{3}_2} + \dots
\end{eqnarray*}
with $\ket{\zeta_k}_{\bar{k}_3}=(\ket{0}_{\bar{k}_3}+e^{-ig_{k3}}\ket{1}_{\bar{k}_3})/\sqrt{2}$ for $k=1,2$ and only the terms surviving in the projections that follow, are shown after the action of unitaries. After projection from virtual to physical qubit in position $3$, we get
\begin{eqnarray*}
    P_3\ket{\bar{\Psi}^\prime}_{N=3} &=&\Big(\ket{0_3}\bra{0_{\bar{3}_1}0_{\bar{3}_2}}+\ket{1_3}\bra{1_{\bar{3}_1}1_{\bar{3}_2}}\Big)\ket{\bar{\Psi}^\prime}_{N=3}\\
    &=&\ket{++}_{\bar{1}_3\bar{2}_3}\ket{0}_3 + \ket{\zeta_1\zeta_2}_{\bar{1}_3\bar{2}_3}\ket{1}_3\\
    &=&\sqrt{2}U_{13}U_{23}\ket{+}_{\bar{1}_3}\ket{+}_{\bar{2}_3}\ket{+}_3,
\end{eqnarray*}
which is not normalized. Therefore, for the $N$-qubit state, the action of projections $P_l$ on qubits in B gives a state (upto normalization) of the form $\ket{\Psi^P}=\bigotimes\limits_{l\in B}^{}\left[\prod_{k\in A}U_{kl} \ket{+}_{\bar{k}_l}\ket{+}_l\right]$ that is a tensor product of $N-n$ states. This allows us to individually trace out each qubit in B.

\subsection{Partial trace over $B$}
Let us consider again the $N=3$ case with $A=\{1\}$ (virtual qubits $\bar{1}_2, \bar{1}_3$) and $B=\{2,3\}$. Then from the previous discussion, $\ket{\Psi^P}=\Big(U_{12}\ket{+}_{\bar{1}_2}\ket{+}_2\Big)\bigotimes\Big(U_{13}\ket{+}_{\bar{1}_3}\ket{+}_3\Big)$ and the partial trace over $\{2,3\}$ can be done independently, by which we obtain $\bar{\rho}_1=\bar{\rho}_1(2)\bigotimes\bar{\rho}_1(3)$, where $\bar{\rho}_1(l)=\ket{+}_{\bar{1}_l}\bra{+}+\ket{\zeta_l}_{\bar{1}_l}\bra{\zeta_l}$, with $\ket{\zeta_l}_{\bar{1}_l}=(\ket{0}_{\bar{1}_l}+e^{-ig_{1l}}\ket{1}_{\bar{1}_l})/\sqrt{2}$ for $l=1,2$.

It is important to note that, although not normalized, each $\bar{\rho}_1(l)$ is a positive semi-definite matrix with all diagonal values as $1$. Generalizing this effect to the $N$-qubit WGS, with bipartition $A$ of $n$ sites, the reduced density matrix is
\begin{equation}
     \bar{\rho}_A = \bigotimes_{l\in B}\bar{\rho}_A(l) = \bigotimes_{l\in B}2^{n-1}\Big( \ket{+}_{\bar{1}_l\bar{2}_l\dots\bar{n}_l}\bra{+}+\ket{\zeta_l}_{\bar{1}_l\bar{2}_l\dots\bar{n}_l}\bra{\zeta_l}\Big),
     \label{eq:rho_l}
\end{equation}
where $\ket{\zeta_l}_{\bar{1}_l\bar{2}_l\dots\bar{n}_l}= \bigotimes_{k=1}^{n}(\ket{0}_{\bar{k}_l}+e^{-ig_{kl}}\ket{1}_{\bar{k}_l})/\sqrt{2}$ for $l\in B$ and each $\bar{\rho}_A(l)$ is a $2^n\times2^n$ positive semi-definite matrix, scaled by $2^{n-1}$ so that all the diagonal values of each $\bar{\rho}_A(l)$ are unity. Note that the weights are encoded in the off-diagonal terms as phases. Proper normalization of $\bar{\rho}_A$ is done only at the end after the projection on virtual qubits in $A$.

\subsection{Action of projection on virtual qubits in $A$}
For the above case of $N=3$ with $A=\{1\}$ (virtual qubits $\bar{1}_2, \bar{1}_3$) and $B=\{2,3\}$,  $\bar{\rho}_1(l)=\ket{0}_{\bar{1}_l}\bra{0}+x_{1l}\ket{1}_{\bar{1}_l}\bra{0}+x^*_{1l}\ket{0}_{\bar{1}_l}\bra{1}+\ket{1}_{\bar{1}_l}\bra{1}$, where $x_{1l}=\frac{1+e^{-ig_{1l}}}{2}$ for $l=2,3$. From the form of the inevitable projection $P_k$ for $k\in A$, the only terms in the tensor product of virtual qubits that will contribute are ones formed only via $\ket{00}_{\bar{1}_2\bar{1}_3}, \ket{11}_{\bar{1}_2\bar{1}_3}$ and their dual vectors. After applying $P_k$, we finally get
\begin{equation}
    \rho^\prime_1=
    \begin{pmatrix}
        1                & x_{12}x_{13}\\
        x^*_{12}x^*_{13} & 1
    \end{pmatrix}
    = \bar{\rho}(2)\ostar\bar{\rho}(3),
\end{equation}
where $\ostar$ is the Hadamard product. This can be generalized for $A$ with $n$ sites, as $\rho^\prime_A=\underset{l\in B}{\ostar}\bar{\rho}_A(l)$, where $\bar{\rho}_A(l)$ is represented as in Eq. (\ref{eq:rho_l}) scaling all diagonal entries to unity. Finally, $\rho^\prime_A$ can be normalized as $\rho_A=\frac{1}{2^n}\rho^\prime_A$.

\section{Efficient calculation of the GGM of a WGS}
\label{sec:WGS_GGM}
\begin{figure}[h]
    \centering
    \includegraphics[width=\linewidth]{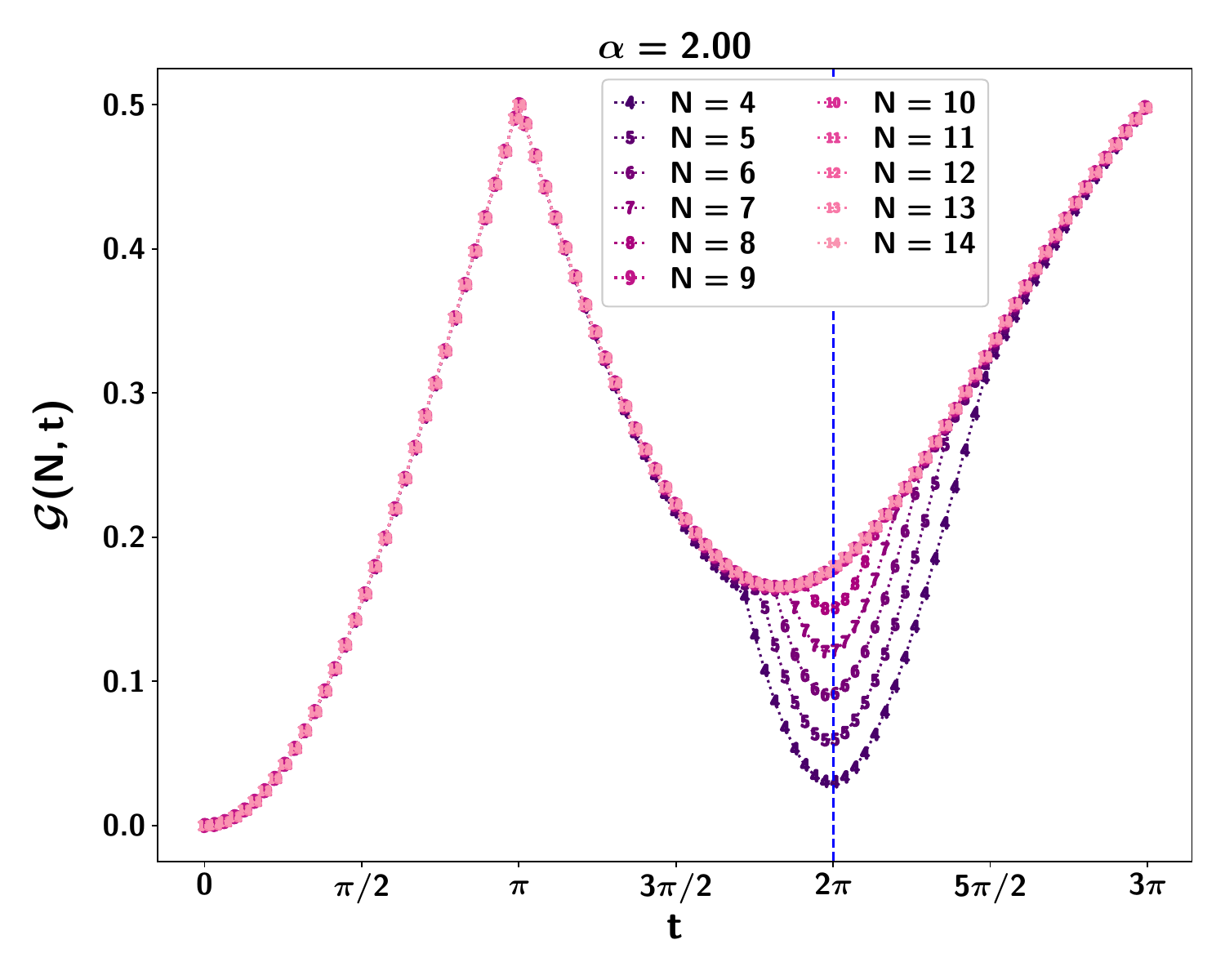}
    \caption{(Color online.) GGM, $\mathcal{G}$ (vertical axis) against time $t$ (horizontal axis) for $N=4$ to $14$ and all-to-all connected WGS ($z=N-1$, $\alpha=2.0$) in the $1$D lattice with open boundary condition. Different lines represent different system sizes $N$. Although $\mathcal{G}$ reaches maximum value $0.5$ at $t=\pi \:\forall N$ (also irrespective of $\alpha$), it depends on system size $N$, mostly around $t=2\pi$ and the dependence of  GGM on system size $N$ disappears with increasing $N$. Both the axes are dimensionless. }
    \label{fig:ggm_smallN}
\end{figure}
\begin{figure}[h]
    \centering
    \includegraphics[width=\linewidth]{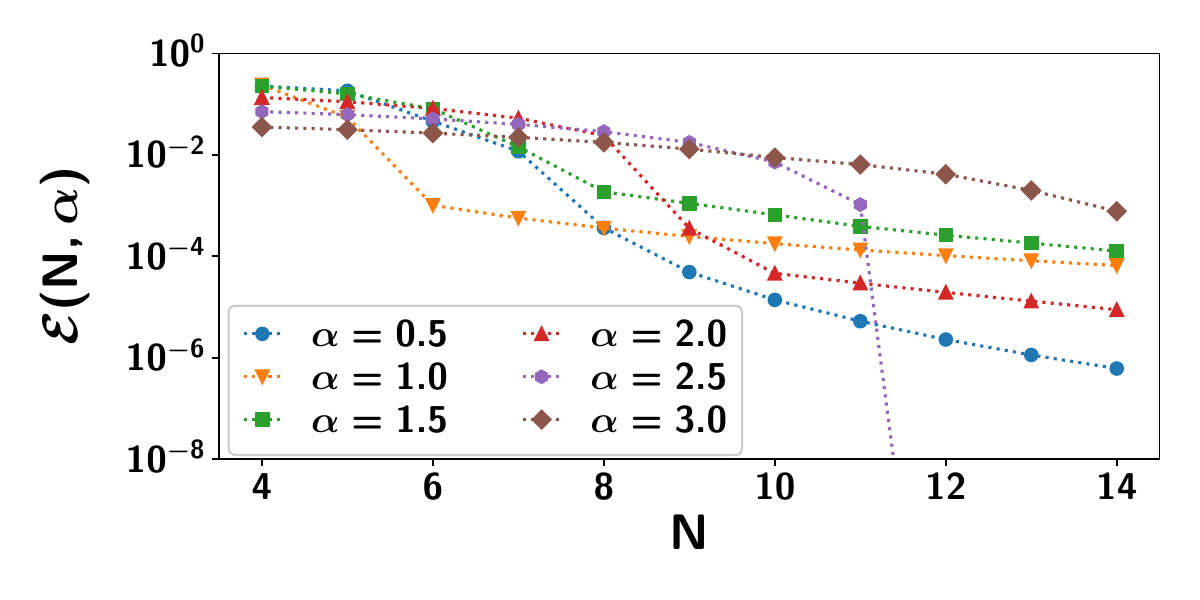}
    \caption{(Color online.) The absolute mean difference $\mathcal{E}$ (ordinate) defined in Eq. (\ref{eq:abs_er}) vs system $N$ (abscissa). As clearly visible,  $\mathcal{E}$ scales exponentially with  $N$, irrespective of \(\alpha\) values. The vertical axis is logarithmic in scale and both the axes are dimensionless. }
    \label{fig:ggm_smallN_err}
\end{figure}

The GGM of a pure state $\ket{\Psi}$ as defined in Eq. (\ref{eq:ggm_def}) becomes computationally hard with increasing system size $N$. The complexity arises because both the number of possible bipartitions and the number of the Schmidt coefficients increase exponentially with the increase of $N$. 

Interestingly, GGM $\mathcal{G}(N,\alpha,t)$ depends on the system size $N$, for small $N$, which is more pronounced around $t=2\pi$ as shown in Fig. \ref{fig:ggm_smallN}. Specifically, $\mathcal{G}(N,\alpha,t)$ shows non-analytic decrease around $t=2\pi$, for small $N$. We numerically find that this is because  bipartition from which the maximum Schmidt coefficient contributes in the computation of GGM changes with the variation of time. For $t<\pi$, $\mathcal{G}$ is always obtained from a single site density matrix, i.e., $|A|=1$. Specifically, the bipartition, $A=\{1\}, B=\{2,3,\dots N\}$ contains the maximum Schmidt coefficient,  as given in Eq. (\ref{eq:single_site}). Therefore, $\mathcal{G}(N,\alpha,t<\pi) = \mathcal{G}_1(N,\alpha,t) = 1-\max\{\omega_1, \omega_2\}$, where $\omega_1$ and $\omega_2$ are the two eigenvalues of $\rho_1$ (i.e., the squares of the corresponding Schmidt coefficient) and the subscript ``\(1\)" in \(\mathcal{G}_1\) indicates that the GGM comes from the eigenvalues of \(\rho_1\).  

Let us now analyze the difference between the  GGM profile obtained by considering all bipartitions and the GGM \(\mathcal{G}_1\), computed only by using the single-site density matrix.  In order to determine this deviation, we define the absolute mean difference as
\begin{equation}
    \mathcal{E}(N,\alpha) = \langle|\mathcal{G}_1(N,\alpha,t)-\mathcal{G}(N,\alpha,t)|\rangle_T,
    \label{eq:abs_er}
\end{equation}
where the average $\langle. \rangle_T$ is taken over time from $t=0$ to $T=3\pi$.
The absolute mean difference $\mathcal{E}(N, \alpha)$, decreases exponentially with increasing system size $N$. For example, the difference is of the order of $10^{-4}$ for $N=14$ for different fall-off rates $\alpha$ as shown in Fig. \ref{fig:ggm_smallN_err}. This shows that for large $N$ GGM can be efficiently computed from the  reduced density  matrix, \(\rho_1\),  independent of the fall-off rates $\alpha$ and time.
This analysis also reveals that the GGM of the weighted graph state with open boundary condition is an edge property. Similar results can be obtained for $2$D \(L\times L\) square lattices, with exact calculations performed for $L=3$ and $4$, in which the GGM is again obtained from the corners of the lattice. 

\bibliography{bib.bib}
\end{document}